\pdfoutput=1

\documentclass[manuscript,acmsmall,screen]{acmart}

\AtBeginDocument{%
  }

\setcopyright{acmlicensed}
\copyrightyear{2018}
\acmYear{2018}
\acmDOI{XXXXXXX.XXXXXXX}

\acmJournal{TECS}
\acmVolume{37}
\acmNumber{4}
\acmArticle{111}
\acmMonth{8}

\usepackage{amsmath,amsfonts}
\usepackage{algorithmic}
\usepackage{algorithm}
\usepackage{array}
\usepackage[caption=false,font=normalsize,labelfont=sf,textfont=sf]{subfig}
\usepackage{textcomp}
\usepackage{stfloats}
\usepackage{url}
\usepackage{verbatim}
\usepackage{graphicx}
\usepackage{booktabs}  
\usepackage{hyperref}  
\usepackage{todonotes}
\usepackage{listings}

\usepackage{xspace}
\usepackage{tikz}
\usetikzlibrary{positioning}
\usetikzlibrary{shapes}

\newtheorem{lemma}{Lemma}
\newtheorem{example}{Example}
\newtheorem{remark}{Remark}

\newcommand{\ctrobust}{\theta}

\lstdefinelanguage{OS2}{
    keywords={scenario, do},
    keywords=[2]{parallel, drive, serial, lane, position},
    keywordstyle={\color{blue}\bfseries},
    keywordstyle=[2]{\color{red!50!blue}},
    sensitive=false, 
    morecomment=[l]{\#}, 
    showstringspaces=false,
    numbers=left,
    stepnumber=1,
} %

\title{Cumulative-Time Signal Temporal Logic}

\author{Hongkai Chen}
\email{hkchen@ie.cuhk.edu.hk}
\affiliation{%
  \institution{The Chinese University of Hong Kong}
  \city{Hong Kong SAR}
  \country{China}
}

\author{Zeyu Zhang}
\affiliation{%
  \institution{Stony Brook University}
  \city{Stony Brook}
  \country{USA}}
\email{zeyu.zhang.2@stonybrook.edu}

\author{Shouvik Roy}
\affiliation{%
  \institution{Illinois Institute of Technology}
  \city{Chicago}
  \country{USA}
}
\email{sroy20@iit.edu}

\author{Ezio Bartocci}
\affiliation{%
 \institution{TU Wien}
 \city{Vienna}
 \country{Austria}}
\email{ezio.bartocci@tuwien.ac.at}

\author{Scott A. Smolka}
\affiliation{%
  \institution{Stony Brook University}
  \city{Stony Brook}
  \country{USA}}
\email{sas@cs.stonybrook.edu}

\author{Scott D. Stoller}
\affiliation{%
  \institution{Stony Brook University}
  \city{Stony Brook}
  \country{USA}}
\email{stoller@cs.stonybrook.edu}

\author{Shan Lin}
\affiliation{%
  \institution{Stony Brook University}
  \city{Stony Brook}
  \country{USA}}
\email{shan.x.lin@stonybrook.edu}

\begin{document}

\begin{abstract}

Signal Temporal Logic (STL) is a widely adopted specification language in cyber-physical systems for expressing critical temporal 
requirements, such as safety conditions and response time.  However, STL's expressivity is not sufficient to capture the cumulative duration during which a property holds within an interval of time. To overcome this limitation, we introduce \emph{Cumulative-Time Signal Temporal Logic} (CT-STL) that operates over discrete-time signals and extends STL with a new \emph{cumulative-time} operator. This operator compares the sum of all time steps for which its nested formula is true with a threshold.  We present both a qualitative and a quantitative (robustness) semantics for CT-STL and prove both their soundness and completeness properties.  We provide an efficient online monitoring algorithm for both semantics. Finally, we show the applicability of CT-STL in two case studies: specifying and monitoring cumulative temporal requirements for a microgrid and an artificial pancreas.

\end{abstract}

\keywords{Cyber-Physical Systems, Signal Temporal Logic, Runtime Verification, Monitoring, Cumulative Temporal Properties}

\maketitle

\section{Introduction}
\label{sec:intro}

The use of formal specification languages and verification
tools is critical in expressing, monitoring, and ensuring important reliability, security, and safety properties in cyber-physical systems~(CPS).  
These systems, which integrate computational algorithms with physical processes, operate in dynamic environments where timing constraints are often as critical as functional correctness.
Signal Temporal Logic (STL)~\cite{maler2004monitoring} is a well-established formalism that allows engineers to specify time-dependent behaviors of real-time systems. 
STL is widely adopted in industry~\cite{MolinANZBE23,Kapinski2016STLibAL,RoehmHM17,chen2022stl} and has been used for a number of verification and validation activities on CPS, including: specification-based monitoring/runtime verification~\cite{BartocciDDFMNS18,JaksicBGKNN15}, control~\cite{chen2023stl}, testing~\cite{BartocciMNY23}, falsification analysis~\cite{FainekosH019}, fault localization~\cite{BartocciFMN18}, and failure explanation~\cite{BartocciMMMN21}.
Despite its versatility, STL lacks the ability to reason about cumulative temporal properties: those that depend on the aggregated duration of events over intervals. 
This limitation becomes particularly acute in applications where accumulated sustained compliance, rather than momentary satisfaction of requirements, determines system safety or efficiency.

Consider a distributed energy resource (DER) integrated into a power grid. Regulatory standards, such as the IEEE Standard 1547-2018~\cite{IEEE1547-2018}, mandate not only that voltage levels avoid instantaneous over-voltage thresholds but also that the cumulative duration of such exceedances remains within strict limits over a time window; e.g., 1.6 ms of overvoltage above 1.7 per unit of nominal instantaneous peak base within a one-minute interval. 
Similarly, in medical devices such as artificial pancreas systems, maintaining blood glucose within safe bounds is insufficient if hypoglycemic episodes, though individually brief, occur too frequently. 
These examples underscore a pervasive challenge: many CPS requirements are inherently cumulative, demanding formalisms that can quantify how long a condition holds---or fails to hold---within a timeframe. 
STL, however, cannot specify requirements involving accumulated durations.

Prior efforts to address this gap have explored augmenting STL’s robustness semantics with algebraic structures like the semiring (min,+)~\cite{JaksicBGN18} or averaging operators~\cite{takumi2015}.
While these approaches enable cumulative reasoning to some extent, they often impose syntactic restrictions~(e.g., requiring formulas in positive normal form) or obscure the intent of specifications with mathematical abstractions~(e.g., convolution kernels~\cite{SilvettiNBB18}).
Such compromises hinder usability, particularly for domain experts unfamiliar with advanced algebraic constructs. 
In contrast, our work introduces a  syntactic extension to STL---a dedicated cumulative-time operator---that preserves STL’s intuitive syntax while enabling direct expression of cumulative requirements.
This operator $\mathbf{C}^\tau_I$   asserts that a formula holds for at least $\tau$ time units within the time interval $I$.

In this paper, we present a novel logic called Cumulative-Time Signal Temporal Logic~(CT-STL) for specifying cumulative-time requirements in cyber-physical systems. 
We extend STL with a cumulative-time operator $\mathbf{C}^\tau_I$ that specifies the cumulative amount of time a property holds in a given time interval $I$ is greater or equal to $\tau$.
Additionally, we define both a qualitative and a quantitative (robustness) semantics for CT-STL and prove the soundness and completeness of the robustness semantics.
We provide an efficient online monitoring algorithm for both semantics. Finally, we demonstrate the utility of CT-STL in two case studies: specifying and monitoring cumulative temporal requirements for a microgrid and an artificial pancreas.
Here, we use an example to illustrate the expressivity of CT-STL.

\begin{figure}[t]
    \centering
    \includegraphics[width=0.5\linewidth]{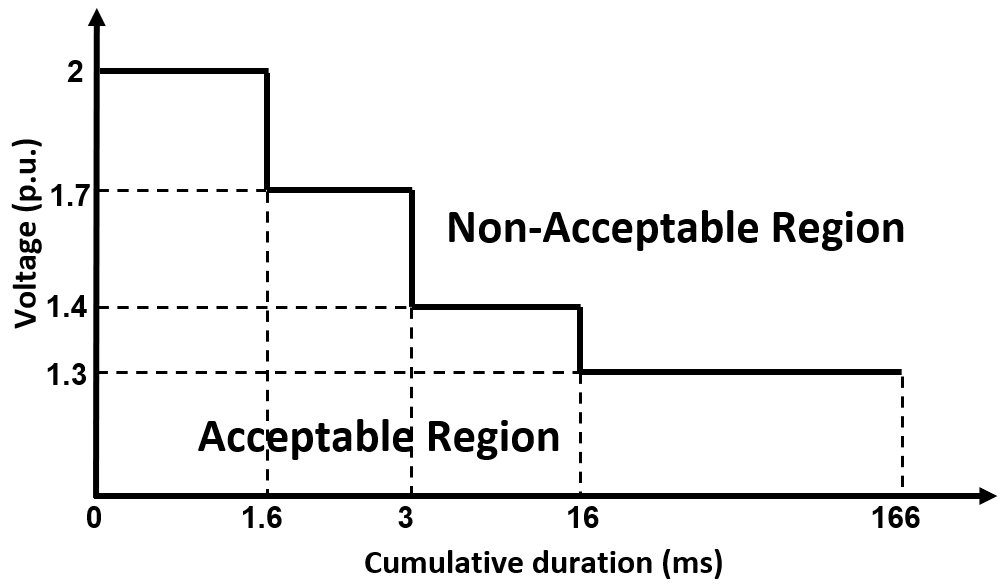}
    \caption{The acceptable region specifies upper bounds for the cumulative time of different overvoltage magnitudes over any one-minute time window.
    This figure is redrawn from the IEEE Standard 1547-2018~\cite{IEEE1547-2018}.}
    \label{fig:ieee-std}
\end{figure}

\begin{example}[IEEE Standard]\label{example:1}
To illustrate CT-STL's expressivity and utility, consider the power quality criteria and requirements for interconnection of distributed energy resources (DERs) with
electric power systems (EPSs) given in the IEEE Standard 1547-2018~\cite{IEEE1547-2018}.
Figure~\ref{fig:ieee-std} illustrates a limitation of overvoltage contribution; i.e., the DER shall not cause the instantaneous voltage on any portion of the Area EPS to exceed the given
magnitudes and cumulative durations. The cumulative time includes the
sum of durations for which the instantaneous voltage exceeds the respective threshold over any one-minute time window. 


In Figure~\ref{fig:ieee-std}, the acceptable region delineates the maximum cumulative durations allowed for various levels of instantaneous overvoltage within any time window of one minute. Specifically, overvoltages exceeding 2.0 per unit (p.u.) of nominal instantaneous peak base are strictly prohibited. Further, the cumulative duration during which the voltage exceeds 1.7 p.u. must not exceed 1.6 ms.
Conversely, this implies that the voltage must remain below 1.7 p.u. for at least (60,000 ms – 1.6 ms) within any one-minute window. Similar cumulative time constraints apply to other threshold levels, forming a descending staircase of allowable overvoltage durations as the voltage magnitude decreases.
Let the step size be $0.1$ ms.
We can specify the acceptable region using a CT-STL formula $\varphi$ as follows.
$$
    \varphi =\; \mathbf{G}_{[0,T]}\; (v \leq 2 
    \;\wedge\; 
    \mathbf{C}_{[0,600000]}^{600000-16}\,v < 1.7 
    \;\wedge\; \mathbf{C}_{[0,600000]}^{600000-30}\,v < 1.4 
    \;\wedge\; \mathbf{C}_{[0,600000]}^{600000-160}\,v < 1.3 )
$$

where $v$ is the absolute value of voltage and $T$ is the operation time of EPSs. The subscript $[0,600000]$ of the \textbf{C} operator represents any time window of $1$ minute ($60000$ ms) during which the subformula is evaluated. Each superscript of the \textbf{C} operator (e.g, $600000-16$) denotes the minimum time the corresponding voltage control (e.g, $v<1.7$) should hold during the time window.

\end{example}

\paragraph{Our contributions} The main contributions of this work are the following:
\begin{itemize}
    \item We propose a new temporal logic called CT-STL that operates on a set of mixed analog signals that we assume to be uniformly sampled. CT-STL extends the STL syntax with a new cumulative time operator $\mathbf{C}^{\tau}_I \varphi$ that sums all of the time durations for which its nested formula $\varphi$ is satisfied over an interval $I$ and compares the result with a threshold of time $\tau$.
    \item We extend the classical notion of STL quantitative semantics (robustness)~\cite{DonzeFM13} to handle the quantitative evaluation of the new cumulative time operator.
    In particular, when a formula includes a cumulative-time operator, the quantitative semantics returns the maximum robustness value of the nested formula within the operator’s time interval.
    
    \item We prove both the soundness and  completeness of the newly defined quantitative semantics.
    \item We provide an online monitoring algorithm for CT-STL.
    \item We demonstrate the utility of our logic and the monitoring algorithm on two case studies.
\end{itemize}

\paragraph{Paper organization} Section~\ref{sec:related-work} considers related work while Section~\ref{sec:background} provides the necessary background.  Sections~\ref{sec:syntax} and~\ref{sec:semantics} define the CT-STL syntax and qualitative/quantitative semantics, respectively.  Section~\ref{sec:sound_complete} provides the proof of  soundness and completeness of CT-STL's quantitative semantics with respect to its qualitative semantics. Section~\ref{sec:monitoring} shows how to extend the online STL monitoring algorithm proposed in~\cite{2017monitor} to handle the cumulative-time operator in CT-STL.  Section~\ref{sec:experiment} presents our two case studies for which the monitoring of cumulative time requirements is critical: a microgrid model subject to a short-circuit simulation and the artificial pancreas systems.  Section~\ref{sec:conclusion} offers our concluding remarks.

\section{Related Work}\label{sec:related-work}

The work by Maler et al.~\cite{maler2004monitoring} introduced Signal Temporal Logic (STL), a dense-time temporal logic that extends a fragment of Metric Interval Temporal Logic~\cite{AlurFH96}.  STL syntax includes predicates over dense-time continuous signals enabling the expression of temporal properties over Boolean and piecewise linear signals.   One of the key features of STL is the possibility of evaluating a set of mixed analog signals against a temporal logic requirement
both qualitatively (it violates/satisfies) and quantitatively~\cite{DonzeFM13} by providing a real value of how much the system satisfies/violates the desired temporal property. 

A main drawback of the classical syntax and semantics of STL is the impossibility of evaluating cumulative temporal properties over time. For example, the standard syntax does not allow expressive cumulative behaviors and the classical quantitative semantic---based on (min,max) semiring---fails to distinguish an isolated single glitch in a signal from more severe, repeated violations. Over the past decade, several works~\cite{zhao2022,haghighi2019control,JaksicBGN18,RodionovaBNG16}  have attempted to address this limitation at the semantic level by replacing the Boolean semiring with alternative algebraic structures, such as, for example, the semiring (min,+).  Haghighi et al.~\cite{haghighi2019control} propose a smooth cumulative robustness designed to efficiently compute control policies. Zhao et al.~\cite{zhao2022} define ASTL, an accumulative robustness metric for STL to support IoT service monitoring throughout the time domain. However, replacing an idempotent operator such as logical $\vee$ with nonidempotent ones such as $+$ introduces challenges. It may impose the limitation of the original STL syntax to \emph{positive normal form}~\cite{RodionovaBNG16} or require special monitoring frameworks~\cite{JaksicBGN18}.  In contrast to these works, our approach operates at STL syntax level: we add a new operator for reasoning about cumulative temporal properties directly in the syntax of the logic, and thus our logic is more expressive than STL.  
Takumi et al.~\cite{takumi2015} propose AvSTL, an extension of the temporal logic of metric intervals by averaged temporal operators. In contrast to~\cite{takumi2015}, our cumulative temporal operator does not average with respect to the time interval, and our notion of robustness provides the largest value of robustness
in the interval of the cumulative time operator
instead of the cumulative robustness. For CT-STL formulas without a cumulative time operator, our quantitative semantics corresponds to the classical STL quantitative semantics. 

A relevant related work is Signal Convolution Logic (SCL) introduced in~\cite{SilvettiNBB18}. SCL replaces the temporal operators with convolution operators with respect to a chosen kernel. In SCL, it is possible to express
cumulative-time properties by choosing the proper kernel. The authors of~\cite{SilvettiNBB18} proved that SCL is actually more expressive than the fragment of STL with only eventually/always operators and without until. The limitations of SCL are that it does not have the possibility to express the temporal until operator (which cannot be expressed with a kernel) and that the user needs to be familiar with the notion of kernel and convolution in order to effectively specify properties.  In contrast, CT-STL has the temporal until operator in its syntax and the cumulative-time operator hides the complexity of using specific kernels.

\section{Background}\label{sec:background}

In this section, we provide background on the syntax and semantics
of STL and its online monitoring algorithms.
Let $\xi:\mathbb{T}\rightarrow\mathbb{R}^n$ be a signal where $\mathbb{T} = \mathbb{Z}_{\geq 0}$ is the (discrete) time domain. 
We describe $\xi$ by a set of $n$ variables $\xi=\{\xi_1,\xi_2,\ldots,\xi_n\}$ and denote by $|\xi|$ the length of $\xi$.

\subsection{Signal Temporal Logic}\label{subsec:stl}

STL is a logic for specifying temporal properties over real-valued signals.
An STL atomic proposition $p\in \mathit{AP}$ is defined over $\xi$ with the form $p\equiv \mu(\xi(t)) \geq c$, 
$c\in\mathbb{R}$, and $\mu:\mathbb{R}^n\rightarrow\mathbb{R}$.  STL formulas $\varphi$ are defined according to the following grammar~\cite{donze2010robust}:
\begin{align*}
    \varphi \,::=\, 
    p\;|\; \neg\,\varphi\;|\; \varphi_1\,\wedge\,\varphi_2 \;|\;     \varphi_1\,\mathbf{U}_I\,\varphi_2 \;
\end{align*}
\noindent where $\mathbf{U}$ is the \emph{until} operator and $I$ is an interval on $\mathbb{T}$. Logical disjunction is derived from logical conjunction $\wedge$ and Boolean negation $\neg$ as usual, and temporal operators \emph{eventually} and \emph{always} are derived from $\mathbf{U}$ as usual: $\mathbf{F}_{I}\varphi = \top\,\mathbf{U}_I\,\varphi$ and $\mathbf{G}_{I}\varphi = \neg\, (\mathbf{F}_I \neg\,\varphi)$.

The satisfaction relation $(\xi,t)\models\varphi$, indicating whether $\xi$ satisfies $\varphi$ at time $t$, is defined as follows:
\begin{align*}
&(\xi, t) \models p &\Leftrightarrow &\hspace{2ex} \mu(\xi(t)) \geq c\\
&(\xi, t) \models \neg \varphi &\Leftrightarrow &\hspace{2ex} \neg ((\xi, t) \models \varphi)\\
&(\xi,t) \models \varphi_1 \wedge \varphi_2  &\Leftrightarrow&\hspace{2ex} (\xi,t) \models \varphi_1 \wedge (\xi,t) \models \varphi_2 \\
&(\xi,t) \models \varphi_1 \mathbf{U}_I \varphi_2 &\Leftrightarrow&\hspace{2ex} \exists\ t'\in t + I\ \text{s.t.}\ (\xi,t') \models \varphi_2 \wedge \forall\; t''\in [t,t'),\ (\xi,t'') \models \varphi_1 
\end{align*}

STL admits a quantitative (robustness) semantics given by a real-valued function $\rho$ such that $\rho(\varphi,\xi,t)>0 \Rightarrow (\xi,t) \models \varphi$, and defined as follows~\cite{donze2010robust}:
\begin{align*}
&\hspace{-3ex}\rho(\mu(\xi(t)) \geq c, \xi, t)&&\hspace{-7ex}=\hspace{2ex} \mu(\xi(t))-c   \\
&\hspace{-3ex}\rho(\neg \varphi, \xi, t)&&\hspace{-7ex}=\hspace{2ex} -\rho(\varphi, \xi, t) \\
&\hspace{-3ex}\rho(\varphi_1 \wedge \varphi_2, \xi, t) &&\hspace{-7ex}=\hspace{2ex} \min(\rho(\varphi_1, \xi, t),\rho(\varphi_2, \xi, t))\\
&\hspace{-3ex}\rho(\varphi_1\mathbf{U}_I\varphi_2, \xi, t) &&\hspace{-7ex}=\hspace{2ex} \max_{t'\in t+I}~\min (\rho(\varphi_2, \xi, t'),\min_{t''\in [t,t+t')}\rho(\varphi_1, \xi, t''))
\end{align*}

\subsection{Online Monitoring for STL}\label{subsec:monitor_stl}

Online monitoring algorithms for STL enable real-time evaluation of temporal requirements over streaming signals, a critical capability for safety-critical CPSs. 
Unlike offline methods that analyze complete traces, online approaches incrementally compute satisfaction or robustness values as partial signal data becomes available. 
STL quantitative semantics, which assigns a robustness degree indicating how much a signal satisfies or violates a specification, quantifies the margin of satisfaction, enabling early termination of monitoring. 
Deshmukh et al.~\cite{2017monitor} proposed online monitoring using robust satisfaction intervals~(RoSI), which track the range of possible robustness values as signals evolve. 
Their approach leverages sliding windows and formula decomposition to incrementally update evaluations, ensuring logarithmic complexity for operators like always and eventually. 
Other monitoring algorithms focused on tasks such as hardware-based monitors~\cite{jakvsic2018quantitative}, predictive monitoring~\cite{lindemann2023conformal}, and runtime enforcement~\cite{sun2024redriver}. 
These works have been evaluated in domains like autonomous systems and medical devices, where online satisfaction monitoring of STL specification can significantly benefit safety. 


\section{Syntax of CT-STL}\label{sec:syntax}

The cumulative-time over a time interval (window) where an STL formula $\varphi$ holds relative to a signal $\xi$ is constrained by some bounds in the example.
In what follows, we introduce the Cumulative-Time Signal Temporal Logic (CT-STL), an extension to STL. 
We define the syntax of CT-STL, describing in detail the cumulative-time operator and CT-STL expressiveness.

An atomic proposition $p$ of Cumulative-Time Signal Temporal Logic (CT-STL) has an identical definition as that of STL~\cite{maler2004monitoring}: $p\in \mathit{AP}$ 
is defined over $\xi$ and is of the form $p\equiv \mu(\xi(t)) \geq c$, 
$c\in\mathbb{R}$, and $\mu:\mathbb{R}^n\rightarrow\mathbb{R}$.
CT-STL formulas $\varphi$ are defined using the following grammar:
\begin{align*}
    \varphi \,::=\, 
    p\;|\; \neg\,\varphi\;|\; 
    \varphi_1\,\wedge\,\varphi_2 \;|\;     \varphi_1\,\mathbf{U}_I\,\varphi_2 \;|\; \mathbf{C}^{\tau}_I \varphi
\end{align*}
\noindent where $I$ is an interval on $\mathbb{T}$, $\tau\in\mathbb{R}_{>0}$, $\tau\leq |I|$, $\mathbf{U}$ is the \emph{until} operator, and $\mathbf{C}$ is the \emph{cumulative-time} operator.
Logical disjunction is derived from $\wedge$ and $\neg$, and operators \emph{eventually} and \emph{always} are derived from $\mathbf{U}$ as in STL: 
${\mathbf{F}}_{I}\varphi =  \top\,{\mathbf{U}}_I\,\varphi, {\mathbf{G}}_{I}\varphi = \neg\, ({\mathbf{F}}_I \neg\,\varphi)$.

The 
operator $\mathbf{C}^{\tau}_I \varphi$ asserts that 
the cumulative amount of time that property $\varphi$ is satisfied over the time interval $I$ is greater than or equal to $\tau$.
We introduce shorthands ${\mathbf{F}}^{\tau}_{I_1,I_2}\varphi = \mathbf{F}_{I_1}\,\mathbf{C}^{\tau}_{I_2}\varphi$, ${\mathbf{G}}^{\tau}_{I_1,I_2}\varphi = \mathbf{G}_{I_1}\,\mathbf{C}^{\tau}_{I_2}\varphi$, and $\varphi_1\,\mathbf{U}^{\tau}_{I_1,I_2}\,\varphi_2 = \varphi_1\,\mathbf{U}_{I_1}\,\mathbf{C}^{\tau}_{I_2}\varphi_2$.
For example, the derived ${\mathbf{F}}^{\tau}_{I_1,I_2}$ operator asserts that there exists a time point in $I_1$ such that the cumulative amount of time property $\varphi$ is satisfied over the subsequent time interval $I_2$ is larger than or equal to $\tau$.


\begin{remark}[Expressiveness of CT-STL]\label{remark:expressiveness}
It is obvious that any STL formula can be expressed in CT-STL, which is an extension of STL.
On the other hand, neither STL, nor any other temporal logic that we are aware of, can express cumulative-time requirements.
The reason is the absence of a temporal operator representing the accumulation of time that some property is satisfied over a time interval. 
For example, let signal $\xi_1$ have positive values
a total of $10$ time units over the time interval $[0, 20]$; $\xi_2$ has a total of $15$ such time units.
STL cannot express or distinguish their difference in such cumulative time: no STL formula definitely holds w.r.t.\ $\xi_1$ but not $\xi_2$ at time $0$.
Therefore, CT-STL is strictly more expressive than STL.
\end{remark}

\begin{remark}[Specifying cumulative-time upper bounds]\label{remark:cumulative_time}
$\mathbf{C}^{\tau}_I \varphi$ asserts that the cumulative amount of time property $\varphi$ is satisfied over the time interval $I$ is greater than or equal to $\tau$.
Hence, $\tau$ is a lower bound on the cumulative time.
A corresponding upper bound $\tau'$ can be specified with the formula
$\mathbf{C}^{|I|-\tau'}_I \neg\varphi$.
Here, we take advantage of the fact that the sum of the cumulative amounts of time a property $\varphi$ and its negation $\neg\varphi$ hold over a time interval $I$ is equal to the length of $I$, denoted $|I|$.
Therefore, an upper bound of 
$\tau'$ for the cumulative amount of time over $I$ that $\varphi$ holds implies a lower bound of $|I|-\tau'$ for the cumulative time over $I$ that $\neg\varphi$ holds.

\end{remark}

In Example~\ref{example:1}, we define cumulative-time upper bounds for overvoltage occurrences by equivalently specifying lower bounds on the cumulative duration during which overvoltage does \emph{not} occur.
This formulation leverages Remark~\ref{remark:cumulative_time}; e.g.,  limiting the cumulative time of $v>1.7$ to less than 16 ms ensures that $v<1.7$ holds for at least (60,000 ms - 16 ms), thus formalized as $\mathbf{C}_{[0,600000]}^{600000-16}\,v < 1.7$.
Next, we demonstrate the utility of the cumulative-time operator in microgrids, where a voltage-related requirement must hold for at least a specified duration to ensure safe operation.

\begin{example}
    Consider a scenario of {load shaving} in a {photovoltaic (PV) system}. During periods of high power generation, the system is prone to {overvoltage}, which can lead to failures in the microgrid. This type of overvoltage, typically within {10\% above the nominal voltage}, is less severe than the case in Example~\ref{example:1}, where the issue is often caused by a microgrid fault.
    A potential mitigation strategy is to employ a Battery Energy Storage System (BESS) to absorb excess power whenever such overvoltage occurs. The BESS continues to charge until the voltage stabilizes.
    
    Let $x_1$ denote the {BESS charging rate} (in kW), and $x_2$ denote the voltage signal. Assuming a time step of 1 second, we can formalize the desired behavior using a CT-STL formula as follows:
    \[
    \varphi = x_1 > 100 \ \mathbf{U}_{[0,7200]} \ \mathbf{C}_{[0,3600]}^{3000} \ x_2 < 1.1
    \]
    This formula specifies that the BESS charging rate should remain above {100 kW}, \emph{until} some time $t$ within the near future (e.g., 2 hours, which is $7200$ seconds), where the cumulative amount of time over the following hour (i.e., $[0, 3600]$) during which the voltage $x_2$ remains below {1.1 p.u.} exceeds {3000 seconds} (i.e., 50 minutes). In other words, the BESS charges to absorb excess power until the voltage is stable for at least 50 minutes in a 1-hour window.
\end{example}
\section{Semantics of CT-STL}\label{sec:semantics}



In this section, we provide both qualitative and quantitative (robustness) semantics for CT-STL.
The satisfaction relation $(\xi,t)\models\varphi$, indicating whether $\xi$ satisfies $\varphi$ at time $t$, is defined 
as follows:

\begin{align*}
&(\xi, t) \models p &\Leftrightarrow &\hspace{2ex} \mu(\xi(t)) \geq c\\
&(\xi, t) \models \neg \varphi &\Leftrightarrow &\hspace{2ex} \neg ((\xi, t) \models \varphi)\\
&(\xi,t) \models \varphi_1 \wedge \varphi_2  &\Leftrightarrow&\hspace{2ex} (\xi,t) \models \varphi_1 \wedge (\xi,t) \models \varphi_2 \\
&(\xi,t) \models \varphi_1 \mathbf{U}_I \varphi_2 &\Leftrightarrow&\hspace{2ex} \exists\ t'\in t + I\ \text{s.t.}\ (\xi,t') \models \varphi_2 \wedge \forall\; t''\in [t,t'),\ (\xi,t'') \models \varphi_1 \\
&(\xi,t) \models \mathbf{C}^{\tau}_I \varphi &\Leftrightarrow&\hspace{2ex} \sum_{t'\in t + I}  [(\xi,t') \models \varphi] \geq \tau
\end{align*}
where function $[(\xi,t)\models\varphi]$ is $1$ if the statement is true, $0$ if the statement is false~\cite{knuth2003concrete}.

The following definition extends the notion of robustness for STL~\cite{donze2010robust} to CT-STL.
Essentially, robustness of a CT-STL formula quantifies the perturbations allowed to the values of atomic propositions evaluated along a signal---and hence, indirectly, the perturbations allowed to the signal itself---such that the Boolean truth value of the CT-STL formula with respect to the signal is preserved.

\begin{definition}[Robustness]\label{def:robustness}
    Let $\varphi$ be a CT-STL formula, $\xi:\mathbb{T}\rightarrow\mathbb{R}^n$ a signal, $c\in\mathbb{R}$, $\mu:\mathbb{R}^n\rightarrow\mathbb{R}$, and $I$ be an interval on $\mathbb{T}$, 
    $\tau\in\mathbb{R}_{>0}$, and $\tau\leq |I|$.
    The 
    robustness $\ctrobust$ of $\varphi$ at time $t$ is defined as follows.
    \begin{align*}
        &\ctrobust(p,\xi,t)&=&\hspace{1em}  \mu(\xi(t)) - c \\
        &\ctrobust(\neg\varphi,\xi,t)&=&\hspace{1em} -\ctrobust(\varphi,\xi,t)\\
        &\ctrobust(\varphi_1\wedge\varphi_2,\xi,t)&=&\hspace{1em}\min(\ctrobust(\varphi_1,\xi,t),\ctrobust(\varphi_2,\xi,t))\\
        &\ctrobust(\varphi_1\mathbf{U}_I\varphi_2,\xi,t)&=&\hspace{1em}\max_{t'\in t+I} \min(\ctrobust(\varphi_2,\xi,t'), \min_{t''\in [t,t')}\ctrobust(\varphi_1,\xi,t''))\\
        &\ctrobust(\mathbf{C}^{\tau}_I \varphi,\xi,t) &=&\hspace{1em}\underset{t'\in t+I}{\textstyle\max^{\tau}}\;\ctrobust(\varphi,\xi,t')
    \end{align*}
where $\max^{\tau}$ is the $\lceil\tau\rceil$-th largest value of a finite real-valued list, which can be found by sorting the list (e.g., using Heapsort, which has a worst-case time complexity of $O(n \log n)$~\cite{bollobas1996best}).
\end{definition}

In particular, the robustness of $\mathbf{C}^\tau_I \varphi$ at time $t$ is defined to be the $\lceil\tau\rceil$-th largest robustness value of $\varphi$ over the time interval $t+I$. 
As a result, perturbations to the values of atomic propositions evaluated along the signal less than $\ctrobust$ do not affect the Boolean truth value of $\varphi$ at the time points at which the robustness is greater than $\ctrobust$.
This ensures the cumulative amount of time for which the robustness value is positive (i.e., property $\varphi$ is satisfied) is larger than or equal to $\tau$; therefore the Boolean truth value of $\mathbf{C}^{\tau}_I \varphi$ is preserved.

\begin{definition}[Secondary signals~\cite{donze2010robust}]
    Let $\xi:\mathbb{T}\rightarrow\mathbb{R}^n$ be a signal and $p = \mu(\xi(t))\geq c\in AP$ an atomic proposition of CT-STL. 
    We call $y(t) = \mu(\xi(t)) - c: \mathbb{T}\rightarrow\mathbb{R}$
    a secondary signal of $\xi$ w.r.t.\ $p$.
\end{definition}

For instance, given a two-dimensional signal $\xi=\{\xi_1,\xi_2\}$ and 
atomic proposition $x_1+2x_2\geq 3$, 
$y(t)=\xi_1(t) + 2\xi_2(t) - 3$ is a secondary signal of $\xi$ w.r.t.\ $x_1+2x_2-3 \geq 0$.

The \textit{characteristic function} $\chi$ of $\varphi$ relative to $\xi$ at time $t$ is defined such that $\chi(\varphi, \xi, t)=1$ when $(\xi,t)\models \varphi$ and $-1$ otherwise.
Its definition in the base case (i.e., for an atomic proposition) is $\chi(p, \xi, t)=sign(\mu(\xi(t)) - c)$, where we interpret $sign(0)$ as $1$~\cite{donze2010robust}.
Let the point-wise distance between two finite signals of the same length be $||y-y'||=\max_t|y[t]-y'[t]|$.
The following theorem about CT-STL robustness is an extension of the results of Fainekos et al.~\cite{fainekos2009robustness} and Donz{\'e} et al.~\cite{donze2010robust}.
{In the former, it is proven that the evaluation of the robust semantics of a Metric Temporal Logic (MTL) formula can be bounded by its robustness
degree, whereas a similar property of STL robustness is presented in the latter without a formal proof. Here, we provide a theorem regarding the CT-STL robustness property, followed by a formal proof.}

\begin{theorem}[Property of Robustness]\label{theorem:preservation}
Let $\varphi$ be a CT-STL formula with a set of atomic propositions $AP=\{p_1,\ldots,p_k\}$;
also, let $\xi:\mathbb{T}\rightarrow\mathbb{R}^n$ be a finite signal, and $t\in\mathbb{T}$. 
For all $\xi'$ in which all secondary signals satisfy $||y_j-{y_j}'||<|\ctrobust(\varphi, \xi, t)|$, $j=1,\ldots,k$,
we have $\chi(\varphi, \xi, t) = \chi(\varphi, \xi', t)$.
\end{theorem}

\begin{proof}
    The proof is by induction on the structure of $\varphi$.
    \begin{itemize}
        \item \textbf{Atomic Propositions} $\varphi=p$: We have $\ctrobust(p,\xi,t)= \mu(\xi(t))-c$.
        By definition, 
        for all $\xi'$ in which all secondary signals satisfy $||y_j-{y_j}'||<|\ctrobust(\varphi, \xi, t)|$, we have 
        \begin{align*}
            |\ctrobust(p,\xi,t)-\ctrobust(p,\xi',t)|& = |\mu(\xi(t))-\mu(\xi'(t))| \\
            &\leq \max_t|\mu(\xi)-\mu(\xi')| \\
            & = ||\mu(\xi)-\mu(\xi')|| 
             = ||y_j-{y_j}'||\\
            & <|\ctrobust(p,\xi,t)|
        \end{align*}
        This suggests 
        $\ctrobust(p,\xi,t)$ and $\ctrobust(p,\xi',t)$ have the same sign.
        Therefore,  $\chi(p, \xi, t) = \chi(p, \xi', t)$.
        \item \textbf{Negation} $\varphi=\neg\varphi_1$: 
        By the induction hypothesis, for all $\xi'$ in which all secondary signals satisfy $||y_j-{y_j}'||<|\ctrobust(\varphi, \xi, t)|=|\ctrobust(\varphi_1, \xi, t)|$, we have $\chi(\varphi_1,\xi,t)=\chi(\varphi_1,\xi',t)$.
        Thus,
        \begin{align*}
-\chi(\varphi_1,\xi,t)=-\chi(\varphi_1,\xi',t)&\Rightarrow\\
\chi(\neg\varphi_1,\xi,t)=\chi(\neg\varphi_1,\xi',t)&\Rightarrow\\
\chi(\varphi,\xi,t)=\chi(\varphi,\xi',t)&
        \end{align*}
        \item \textbf{Conjunction} $\varphi=\varphi_1\wedge\varphi_2$: 
        Without loss of generality, we assume $\ctrobust(\varphi_1, \xi, t)\leq \ctrobust(\varphi_2, \xi, t)$.
        By Definition~\ref{def:robustness} we have 
        \begin{align*}
        \ctrobust(\varphi_1\wedge\varphi_2,\xi,t)=\min\{\ctrobust(\varphi_1,\xi,t),\ctrobust(\varphi_2,\xi,t)\} = \ctrobust(\varphi_1, \xi, t)
        \end{align*}
        1)~If $0\leq \ctrobust(\varphi_1, \xi, t)$, for all $\xi'$ in which all secondary signals satisfy
        \begin{align*}
            ||y_j-{y}'_{j}||<|\ctrobust(\varphi,\xi,t)|=|\ctrobust(\varphi_1, \xi, t)| \leq |\ctrobust(\varphi_2, \xi, t)|,
        \end{align*}
        by the induction hypothesis, we have
        \begin{align*}
            \chi(\varphi_1,\xi,t) & = \chi(\varphi_1,\xi',t)\\
            \chi(\varphi_2,\xi,t) & = \chi(\varphi_2,\xi',t)
        \end{align*}
        Therefore, it can be easily shown that
        \begin{align*}
            \chi(\varphi_1\wedge\varphi_2,\xi,t)&= \min(\chi(\varphi_1,\xi,t),\chi(\varphi_2,\xi,t))\\
            &= \min(\chi(\varphi_1,\xi',t),\chi(\varphi_2,\xi',t))\\
            &= \chi(\varphi_1\wedge\varphi_2,\xi',t)
        \end{align*}
        2)~If $\ctrobust(\varphi_1, \xi, t) < 0$, for all $\xi'$ in which all secondary signals satisfy $||y_j-{y}'_{j}||<|\ctrobust(\varphi,\xi,t)|=|\ctrobust(\varphi_1, \xi, t)|$, by the induction hypothesis, we have
        \begin{align*}
\chi(\varphi_1,\xi,t) & = \chi(\varphi_1,\xi',t) = -1
        \end{align*}
        It can be easily shown that
        \begin{align*}
        \chi(\varphi_1\wedge\varphi_2,\xi,t)= \min(\chi(\varphi_1,\xi,t),\chi(\varphi_2,\xi,t)) = -1
        \end{align*}
        and
        \begin{align*}
        \chi(\varphi_1\wedge\varphi_2,\xi',t)= \min(\chi(\varphi_1,\xi',t),\chi(\varphi_2,\xi',t))=-1
        \end{align*}
        Therefore, 
        \begin{align*}
            \chi(\varphi_1\wedge\varphi_2,\xi,t)=\chi(\varphi_1\wedge\varphi_2,\xi',t)
        \end{align*}
        \item \textbf{Until} $\varphi=\varphi_1\mathbf{U}_I\varphi_2$: By Definition~\ref{def:robustness}, we have
        \begin{align*}
            \ctrobust(\varphi_1\mathbf{U}_I\varphi_2,\xi,t)=\max_{t'\in t+I} \min(\ctrobust(\varphi_2,\xi,t'), \min_{t''\in [t,t')}\ctrobust(\varphi_1,\xi,t''))
        \end{align*}
        1)~If $0\leq\ctrobust(\varphi_1\mathbf{U}_I\varphi_2,\xi,t)$, due to the maximum operator in the robustness definition, we have $\exists t'\in t+I$ such that 
        \begin{align*}
            \min(\ctrobust(\varphi_2,\xi,t'), \min_{t''\in [t,t')}\ctrobust(\varphi_1,\xi,t'')) = \ctrobust(\varphi_1\mathbf{U}_I\varphi_2,\xi,t) 
        \end{align*}
        Due to the minimum operators, we have
        \begin{align*}
            \ctrobust(\varphi_2,\xi,t') \geq \ctrobust(\varphi_1\mathbf{U}_I\varphi_2,\xi,t)\geq 0, &\\
            \forall t''\in[t,t'), \ctrobust(\varphi_1,\xi,t'') \geq  \ctrobust(\varphi_1\mathbf{U}_I\varphi_2,\xi,t)\geq 0
        \end{align*}
        For all $\xi'$ in which all secondary signals satisfy $||y_j-{y}'_{j}||<|\ctrobust(\varphi_1\mathbf{U}_I\varphi_2,\xi,t)|$, by the induction hypothesis, we have $\exists t'\in t+I$ such that
        \begin{align*}
            \chi(\varphi_2,\xi',t')&=\chi(\varphi_2,\xi,t') = 1\\
            \forall t''\in[t,t'), \chi(\varphi_1,\xi',t'')&=\chi(\varphi_1,\xi,t'')=1
        \end{align*}
        Thus, it follows that
        \begin{align*}
            \chi(\varphi_1\mathbf{U}_I\varphi_2,\xi,t)
            = \chi(\varphi_1\mathbf{U}_I\varphi_2,\xi',t) =1
        \end{align*}
        2)~If $\ctrobust(\varphi_1\mathbf{U}_I\varphi_2,\xi,t)<0$, due to the maximum operator in the robustness definition, we have $\forall t'\in t+I$, 
        \begin{align*}
            \min(\ctrobust(\varphi_2,\xi,t'), \min_{t''\in [t,t')}\ctrobust(\varphi_1,\xi,t'')) \leq  \ctrobust(\varphi_1\mathbf{U}_I\varphi_2,\xi,t) < 0
        \end{align*}
        Due to the minimum operators, we have
        \begin{align*}
            \ctrobust(\varphi_2,\xi,t') \leq \ctrobust(\varphi_1\mathbf{U}_I\varphi_2,\xi,t)< 0, 
        \end{align*}
        or
        \begin{align*}
            \exists t''\in[t,t'), s.t.\ \ctrobust(\varphi_1,\xi,t'') \leq  \ctrobust(\varphi_1\mathbf{U}_I\varphi_2,\xi,t)< 0
        \end{align*}
        For all $\xi'$ in which all secondary signals satisfy $||y_j-{y}'_{j}||<|\ctrobust(\varphi_1\mathbf{U}_I\varphi_2,\xi,t)|$, by the induction hypothesis, we have  $\forall t'\in t+I$, 
        \begin{align*}
            \chi(\varphi_2,\xi',t') = \chi(\varphi_2,\xi,t)=-1, 
        \end{align*}
        or
        \begin{align*}
            \exists t''\in[t,t'), s.t.\ \chi(\varphi_1,\xi',t'') =  \chi(\varphi_1,\xi,t) = -1
        \end{align*}
        Thus, it can be easily shown that
        \begin{align*}
            \chi(\varphi_1\mathbf{U}_I\varphi_2,\xi,t)
            = \chi(\varphi_1\mathbf{U}_I\varphi_2,\xi',t) = -1
        \end{align*}
        \item \textbf{Cumulative-time} $\varphi=\mathbf{C}^{\tau}_I \varphi_1$: 
        By Definition~\ref{def:robustness}, we have
        \begin{align*}
            \ctrobust(\mathbf{C}^{\tau}_I \varphi_1,\xi,t)=\underset{t'\in t+I}{\textstyle\max^{\tau}}\;\ctrobust(\varphi_1,\xi,t')
        \end{align*}
        1)~If $0\leq\ctrobust(\mathbf{C}^{\tau}_I \varphi_1,\xi,t)$, it implies that the cumulative time of $|\ctrobust(\varphi_1,\xi,t')|\geq|\ctrobust(\mathbf{C}^{\tau}_I \varphi_1,\xi,t)|$ over time interval $I$ is at least $\lceil\tau\rceil$, i.e.,
        \begin{align*}
        \sum_{t'\in t + I}  [|\ctrobust(\varphi_1,\xi,t')|\geq|\ctrobust(\mathbf{C}^{\tau}_I \varphi_1,\xi,t)|] \geq \lceil\tau\rceil
        \end{align*}
        For all $\xi'$ in which all secondary signals satisfy $||y_j-{y}'_{j}||<|\ctrobust(\mathbf{C}^{\tau}_I \varphi_1,\xi,t)|$, by the induction hypothesis, we have
        \begin{align*}
         \sum_{t'\in t + I}  [\chi(\varphi_1,\xi',t') = \chi(\varphi_1,\xi,t')=1] \geq \lceil\tau\rceil
        \end{align*}       
        Thus, 
        \begin{align*}
            \chi(\mathbf{C}^{\tau}_I \varphi_1,\xi,t) = \chi(\mathbf{C}^{\tau}_I \varphi_1,\xi',t) = 1
        \end{align*}
        2)~If $\ctrobust(\mathbf{C}^{\tau}_I \varphi_1,\xi,t)<0$, it implies that the cumulative time of $|\ctrobust(\varphi_1,\xi,t')|\geq |\ctrobust(\mathbf{C}^{\tau}_I \varphi_1,\xi,t)|$ over time interval $I$ is at least $|I|-\lceil\tau\rceil+1$, i.e.,
        \begin{align*}
                \sum_{t'\in t + I}  [|\ctrobust(\varphi_1,\xi,t')|\geq |\ctrobust(\mathbf{C}^{\tau}_I \varphi_1,\xi,t)|] \geq |I|-\lceil\tau\rceil+1
        \end{align*}
        For all $\xi'$ in which all secondary signals satisfy $||y_j-{y}'_{j}||<|\ctrobust(\mathbf{C}^{\tau}_I \varphi_1,\xi,t)|$, by the induction hypothesis, we have
        \begin{align*}
            \sum_{t'\in t + I}  [ \chi(\varphi_1,\xi',t') = \chi(\varphi_1,\xi,t')=-1] \geq |I|-\lceil\tau\rceil+1
        \end{align*}
        Thus, 
        \begin{align*}
            \chi(\mathbf{C}^{\tau}_I \varphi_1,\xi,t) = \chi(\mathbf{C}^{\tau}_I \varphi_1,\xi',t) = -1
        \end{align*}
    \end{itemize}
    This completes the proof.
\end{proof}

\begin{example}
    This example illustrate Theorem~\ref{theorem:preservation} for the  $\mathbf{C}$ operator.
    Consider $\varphi = \mathbf{C}^{4}_{[2,8]} x> 1$ and discrete-time signal $x_1$, at time $t=0,\ldots,10$, shown in the table below.
    \begin{table}[h]
\centering
\resizebox{0.4\textwidth}{!}{%
\begin{tabular}{@{}llllllllllll@{}}
\toprule
$t$  & 0 & 1 & 2 & 3 & 4 & 5 & 6 & 7 & 8 & 9 & 10 \\ \midrule
$x_1$ & 0  & 0  & 2  & 3  & 4  & 7  & 10  & 0 & 5  & 5 & 15 \\ \midrule
${\chi(x>1,x_1,t)}$  & -1  & -1  & 1  & 1  & 1  & 1  &  1 & -1  & 1  & 1  & 1  \\ \midrule
$\ctrobust(x>1,x_1,t)$  & -1  & -1  & 1  & 2  & 3  & 6  & 9 & -1  & 4  & 4  & 14  \\ 
\bottomrule
\end{tabular}%
}
\end{table}

The cumulative time of $x_1$ such that $x> 1$ over $[2,8]$ is $6$ (at time $t=2,3,4,5,6,8$), so $(x_1,t) \models \varphi$.
By Definition~\ref{def:robustness}, we have 
\begin{align*}
    \ctrobust(\mathbf{C}^{4}_{[2,8]} x> 1,x_1,0) = {\textstyle\max^{4}}\;\{1,2,3,6,9,-1,4\} =4
\end{align*}
Thus, if all changes to $x_1$ are less than $4$, then the truth value of $\varphi$ is preserved, i.e., $(x_1,t) \models \varphi$ ($x>1$ holds at a minimum of 4 time points, namely, $t=4,5,6,8$).

Consider signal $x_2$, shown in the table below, in which the cumulative time 
such that $x>1$ over $[2,8]$ is $2$ (at time $t=2,5$), therefore $(x_1,t) \not\models \varphi$.
\begin{table}[h]
\centering
\resizebox{0.4\textwidth}{!}{%
\begin{tabular}{@{}llllllllllll@{}}
\toprule
$t$  & 0 & 1 & 2 & 3 & 4 & 5 & 6 & 7 & 8 & 9 & 10 \\ \midrule
$x_2$ & 0  & 0  & 2  & -3  & -4  & 7  & -5  & -1 & 0  & 5 & 15 \\ \midrule
${\chi(x>1,x_1,t)}$  & -1  & -1  & 1  & -1  & -1  & 1  &  -1 & -1  & -1  & 1  & 1  \\ \midrule
$\ctrobust(x>1,x_1,t)$  & -1  & -1  & 1  & -4  & -5  & 6  & -6 & -2  & -1  & 4  & 14  \\ 
\bottomrule
\end{tabular}
}
\end{table}

By Definition~\ref{def:robustness}, we have 
\begin{align*}
    \ctrobust(\mathbf{C}^{4}_{[2,8]} x> 1,x_2,0)={\textstyle\max^{4}}\;\{1,-4,-5,6,-6,-2,-1\} = -2
\end{align*}
Thus, if all changes to $x_2$ are less than $|-2|$, then the truth value of $\varphi$ is preserved, i.e., $(x_2,t) \not\models \varphi$. ($x>1$ does not hold at a minimum of 4 time points, namely, $t=3,4,6,7$).
\end{example}

\begin{remark}
    If signals have a uniform time-step size of an arbitrary $\delta\in \mathbb{R}$, function $\max^{\tau}$ in Definition~\ref{def:robustness} 
    should be replaced by $\max^{\tau/\delta}$.
    The intuition is that
    if $\ctrobust(\varphi, \xi, t)> 0$ holds at more than $\lceil\tau/\delta\rceil$ time steps, the cumulative time where $(\xi,t) \models \varphi$ holds is more than $\tau$.
\end{remark}

     

\section{Soundness and Completeness}\label{sec:sound_complete}

In this section, we prove the soundness and completeness of CT-STL cumulative-time robustness.

\begin{theorem}\label{theorem:soundness}
Let $\xi$ be a signal and $\varphi$ a CT-STL specification. 
The following results at time $t$ hold:
\begin{enumerate}
    \item $\ctrobust(\varphi,\xi,t) >0 \Leftrightarrow (\xi,t)\models \varphi$
    \item $\ctrobust(\varphi,\xi,t) < 0 \Leftrightarrow (\xi,t)\models \neg\varphi$
\end{enumerate}
\end{theorem}

\begin{proof}
We prove statement~(1) of Theorem~\ref{theorem:soundness}. Statement~(2) can be proven analogously.
The proofs follow the inductive structure used in the cumulative-time robustness definition.


($\Rightarrow$)

\begin{itemize}
    \item $\varphi = p$. By definition, $\ctrobust(p,\xi,t) = \mu(\xi(t)) - c > 0$, which implies $(\xi,t) \models p$.

    \item $\varphi = \neg\varphi_1$. By definition, $\ctrobust(\varphi,\xi,t) = -\ctrobust(\varphi_1,\xi,t) > 0$. 
    
    By the induction hypothesis, $\ctrobust(\varphi_1,\xi,t) < 0$ implies $(\xi,t) \models \neg\varphi_1$, so $(\xi,t) \models \varphi$.

    \item $\varphi = \varphi_1 \wedge \varphi_2$. By definition, $\ctrobust(\varphi,\xi,t) = \min(\ctrobust(\varphi_1,\xi,t), \ctrobust(\varphi_2,\xi,t)) > 0$, so both $$\ctrobust(\varphi_1,\xi,t) > 0$$ and $$\ctrobust(\varphi_2,\xi,t) > 0$$
    
    Hence, $(\xi,t) \models \varphi_1$ and $(\xi,t) \models \varphi_2$, which implies $(\xi,t) \models \varphi$.

    \item $\varphi = \varphi_1 \mathbf{U}_I \varphi_2$. By definition, $\ctrobust(\varphi,\xi,t) = \max_{t' \in t+I} \min\big(\ctrobust(\varphi_2,\xi,t'), \min_{t'' \in [t,t')} \ctrobust(\varphi_1,\xi,t'')\big) > 0$. Therefore, there exists $t' \in t+I$ such that $\ctrobust(\varphi_2,\xi,t') > 0$ and $\ctrobust(\varphi_1,\xi,t'') > 0$ for all $t'' \in [t,t')$. By the induction hypothesis, $(\xi,t'') \models \varphi_1$ for all such $t''$ and $(\xi,t') \models \varphi_2$, so $(\xi,t) \models \varphi$.
    \item $\varphi = \mathbf{C}_I^{\tau} \varphi_1$. By definition, $\ctrobust(\varphi,\xi,t) = \max^{\tau}_{t' \in t+I} \ctrobust(\varphi_1,\xi,t') > 0$, which implies $$\sum_{t' \in t+I} [\ctrobust(\varphi_1,\xi,t') > 0] \geq \tau$$ From the induction hypothesis, this gives $\sum_{t' \in t+I} [(\xi,t') \models \varphi_1] \geq \tau$, hence $(\xi,t) \models \varphi$.
\end{itemize}

($\Leftarrow$)
\begin{itemize}
    \item $\varphi = p$. By definition, $(\xi,t)\models p$ implies $\mu(\xi(t)) > c$, which then implies $\ctrobust(p,\xi,t) =\mu(\xi(t)) - c > 0$.
    \item $\varphi =\neg\varphi_1$. 
    By definition, $(\xi,t)\models\neg\varphi_1$ implies that $\neg((\xi,t)\models\varphi_1)$.
    From the induction hypothesis, we have $\ctrobust(\varphi_1,\xi,t) < 0$, which implies $\ctrobust(\varphi,\xi,t)= -\ctrobust(\varphi_1,\xi,t) > 0$.
    \item $\varphi =\varphi_1 \wedge \varphi_2$.
    This is equivalent to $(\xi,t)\models \varphi_1 \wedge (\xi,t)\models \varphi_2$ holds.
    From induction hypothesis, we have $\ctrobust(\varphi_1,\xi,t) > 0$; and $\ctrobust(\varphi_2,\xi,t) > 0$.
    By Definition~\ref{def:robustness}, we conclude that $\ctrobust(\varphi,\xi,t) = \min(\ctrobust(\varphi_1,\xi,t),\ctrobust(\varphi_2,\xi,t)) > 0$.
    \item $\varphi = \varphi_1 \mathbf{U}_I \varphi_2$. By definition, $(\xi,t) \models \varphi_1 \mathbf{U}_I \varphi_2$ implies that there exists $t' \in t + I$ such that $(\xi,t') \models \varphi_2$ and $(\xi,t'') \models \varphi_1$ for all $t'' \in [t,t')$. By the induction hypothesis, this gives $\ctrobust(\varphi_2,\xi,t') > 0$ and $\ctrobust(\varphi_1,\xi,t'') > 0$ for all $t'' \in [t,t')$. 
    
    Therefore, $$\min\big(\ctrobust(\varphi_2,\xi,t'),\, \min_{t'' \in [t,t')} \ctrobust(\varphi_1,\xi,t'')\big) > 0$$
    and hence $\ctrobust(\varphi,\xi,t) = \max_{t' \in t+I} \min\big(\ctrobust(\varphi_2,\xi,t'),\, \min_{t'' \in [t,t')} \ctrobust(\varphi_1,\xi,t'')\big) > 0$.
    \item $\varphi = \mathbf{C}^{\tau}_I \varphi_1$. 
    This implies that $\sum_{t'\in t + I}  [(\xi,t')\models \varphi_1] \geq \tau$.
    From the induction hypothesis, $\sum_{t'\in t + I}  [\ctrobust(\varphi,\xi,t') >0] \geq \tau$.
    Thus, $\ctrobust(\varphi,\xi,t)=\underset{t'\in t+I}{\textstyle\max^{\tau}}\;\ctrobust(\varphi,\xi,t') >0$.
\end{itemize}

This completes the proof.
\end{proof}
\renewcommand{\algorithmicrequire}{\textbf{Input:}}
\renewcommand{\algorithmicensure}{\textbf{Output:}}
\newcommand{\SWITCH}[1]{\STATE \textbf{switch} (#1)}
\newcommand{\ENDSWITCH}{\STATE \textbf{end switch}}
\newcommand{\CASE}[1]{\STATE \textbf{case} #1\textbf{:} \begin{ALC@g}}
\newcommand{\ENDCASE}{\end{ALC@g}}
\newcommand{\CASELINE}[1]{\STATE \textbf{case} #1\textbf{:} }
\newcommand{\DEFAULT}{\STATE \textbf{default:} \begin{ALC@g}}
\newcommand{\ENDDEFAULT}{\end{ALC@g}}
\newcommand{\DEFAULTLINE}[1]{\STATE \textbf{default:} }

\section{Monitoring CT-STL}\label{sec:monitoring}

In this section, we review Deshmukh et al.'s \cite{2017monitor} online monitoring method for standard STL and extend it to monitor CT-STL.  As in~\cite{2017monitor},
we assume we monitor a signal over a finite time domain.



\begin{definition}[Prefix and completion~\cite{2017monitor}]\label{def:completions}
    Let $\mathbb{T'}=\{0,\ldots,i\} \subseteq \mathbb{T}$ and $\xi_{[0,i]}$ be a partial signal over $\mathbb{T'}$.
    Then $\xi_{[0,i]}$ is a {prefix} of signal $\xi$ if for all $0 \leq t \leq i$, $\xi(t)=\xi_{[0,i]}(t)$; a signal $\xi$ is a completion of $\xi_{[0,i]}$ if $\xi_{[0,i]}$ is a prefix of $\xi$.  
    Let $S(\xi_{[0,i]})$ denote the set of completions of $\xi$.
\end{definition}

\begin{definition}[Robust satisfaction interval (RoSI)]\label{def:rosi}
    We extend the RoSI defined in~\cite{2017monitor} and define the RoSI of an CT-STL formula $\varphi$ on a partial signal $\xi_{[0,i]}$ at a time $t$ as an interval $I$ such that:
    \begin{align}
        \text{inf}(I) &= \underset{\xi \in S(\xi_{[0,i]})}{\text{inf}} \ctrobust (\varphi, \xi, t)\\
        \text{sup}(I) &= \underset{\xi \in S(\xi_{[0,i]})}{\text{sup}} \ctrobust (\varphi, \xi, t)
    \end{align}
\end{definition}

Next, we define a function $[\theta]$ and then show that it computes the RoSI.  This is similar to \cite{2017monitor}, except we add equation (\ref{eq-rosi-C}) for the $\mathbf{C}$ operator, and we omit the cases for Eventually and Always, because we treat them as derived operators.
Let interval $I_j=(a_j, b_j)$, $a_j$, $b_j \in \mathbb{R}$.
we define $\underset{j}{\max}^{\tau}(I_j) = (\underset{j}{\max}^{\tau}(a_j), \underset{j}{\max}^{\tau}(b_j))$, where the operations on the right-hand side are defined in Definition~\ref{def:robustness}.

\begin{definition}
    We define a recursive function $[\ctrobust]$ that maps a given CT-STL formula $\varphi$, a partial signal $\xi_{[0, i]}$ and a time $t$ to an interval $[\ctrobust](\varphi, \xi_{[0,i]}, t)$. The notation $\mu_{\text{inf}}$ and $\mu_{\text{sup}}$  respectively denote the infimal and supremal value of the function $\mu(\xi)$ over the signal domain $\mathbb{R}^n$, respectively.

    \begin{align}
        [\ctrobust](\mu(\xi_{[0,i]}) \geq 0, \xi_{[0,i]}, t) &= \left\{
        \begin{aligned}
            &[\mu(\xi_{[0,i]}(t)),\mu(\xi_{[0,i]}(t))], 
            &&t \in \mathbb{T}'\\
            &[\mu_{\text{inf}}, \mu_{\text{sup}}], 
            && \text{otherwise}
        \end{aligned}
        \right.\\
        [\ctrobust](\neg \varphi, \xi_{[0,i]}, t) &= -[\ctrobust](\varphi, \xi_{[0,i]}, t)\\
        [\ctrobust](\varphi_1 \wedge \varphi_2, \xi_{[0,i]}, t) &=\min ([\ctrobust](\varphi_1, \xi_{[0,i]}, t), [\ctrobust](\varphi_2, \xi_{[0,i]}, t)) \label{eq-rosi-and}\\
        [\ctrobust](\varphi_1 \mathbf{U}_I \varphi_2, \xi_{[0,i]}, t) &= \sup_{t_1\in t+I} \min (
            [\ctrobust](\varphi_2, \xi_{[0,i]}, t_1),
            \inf_{t_2\in[t, t_1]}(\varphi_1, \xi_{[0,i]}, t_2) ) \label{eq-rosi-u}\\
        [\ctrobust](\mathbf{C}_I^{\tau} \varphi, \xi_{[0,i]}, t) &= \underset{t'\in t+I}{\text{max}^{\tau}} ([\ctrobust](\varphi, \xi_{[0,i]}, t')) \label{eq-rosi-C}
    \end{align}
\end{definition}

Equations (\ref{eq-rosi-and})--(\ref{eq-rosi-u}) use interval arithmetic operations introduced in ~\cite{2017monitor}. 
Now we show that $[\theta]$ calculates the RoSI of CT-STL.
\begin{lemma}
     For any CT-STL formula $\varphi$, the function $[\ctrobust](\varphi, \xi_{[0,i]}, t)$ defines the RoSI of $\varphi$ over the partial signal $\xi_{[0,i]}$ at time $t$. 
\end{lemma}
   
\begin{proof}
The proof is by induction over the structure of the formula.  The correctness of 
the equations for the STL operators follows from~\cite{2017monitor}. Here we prove the correctness of equation (\ref{eq-rosi-C}). 
Let $I=[a,b]$ and 
$[\ctrobust](\varphi, \xi_{[0,i]}, t+i) = [p_i, q_i], i\in I$.
Assume  $[\ctrobust](\varphi, \xi_{[0,i]}, t)$ is a RoSI, i.e., for any completion $\xi$ of $\xi_{[0,i]}$, $\ctrobust(\varphi, \xi, t)\in [\ctrobust](\varphi, \xi_{[0,i]}, t)$. 
By definition, $\ctrobust(\mathbf{C}_I^{\tau} \varphi, \xi_{[0,i]}, t) = \underset{t'\in t+I}{\textstyle\max^{\tau}}\;\ctrobust(\varphi,\xi,t')$ is the $\tau$-th largest value of $\ctrobust(\varphi,\xi,t'), t'\in t+I$.
It is obvious that $\tau$-th largest value is less than the $\tau$-th largest upper bound; similarly, it is greater than $\tau$-th largest lower bound, i.e., $$\underset{t'\in t+I}{\textstyle\max^{\tau}}\; p_i \leq \underset{t'\in t+I}{\textstyle\max^{\tau}}\;\ctrobust(\varphi, \xi_{[0,i]}, t) \leq \underset{t'\in t+I}{\textstyle\max^{\tau}} q_i$$
This concludes the proof.
\end{proof}

\begin{algorithm}[tb]
    \caption{updateWorkList($v_{\psi}, i+1, \xi(i+1)$)}
    \label{alg:updateworklist}
 \begin{algorithmic}[1]
    \SWITCH {$\psi$} \textbf{do} 
        \CASE{$\mathbf{C}_I^{\tau} \varphi$}
            \STATE updateWorkList($v_{\varphi}, i+1, \xi(i+1)$)
            \STATE worklist[$v_{\psi}$] = SlidingC($\mathbf{C}_I^{\tau} \xi>0$, worklist[$v_{\varphi}$], $t$) 
        \ENDCASE
        \STATE \textbf{other cases:} \\
        \STATE  \quad  Run Algorithm 2: updateWorkList($v_{\varphi}, i+1, \xi(i+1)$) in~\cite{2017monitor}.
    \ENDSWITCH
\end{algorithmic}
\end{algorithm}

We follow the online monitoring approach proposed in \cite{2017monitor}. Although they focus on non-uniformly sampled signals, i.e., time steps are not uniform, it is easy to convert to uniformly sampled signals with unit time step.  We briefly summarize their approach; see \cite{2017monitor} for additional explanation. The algorithm works bottom-up on the syntax tree of the monitored formula.  It computes the RoSI for the formula associated with each node $v$ by combining the RoSIs for its sub-formulas, i.e., its children.  The computation of the RoSI at a node $v$ has to be delayed until the RoSIs at its children in a certain time interval are available. This time interval is called the time horizon of $v$ and denoted hor($v$).  If $v$ is the root node, hor($v$)=[0]. Otherwise, if the parent of $v$ is a temporal operator $\mathbf{H}_I$ $(\mathbf{H}=\mathbf{F},\mathbf{G},\mathbf{U}$), then hor($v$)=$I$+hor(parent($v$)). If the parent of $v$ is a logic operator, then hor($v$)=hor(parent($v$)). 

For every node, there is a worklist that stores the RoSI of the associated formula at multiple time instants.  Let $\psi$ denote the subformula corresponding to node $v$. Specifically, after processing a partial signal
$\xi_{[0,i]}$, the algorithm maintains in worklist[$v$]($t$) the RoSI $[\theta](\psi, \xi_{[0,i]}, t)$ for each $t \in {\rm hor}(v) \cap [0, i]$.
Algorithm~\ref{alg:updateworklist} updates the RoSI at time $t$ in the worklist of $v$ from $[\theta](\psi, \xi_{[0,i]}, t)$ to $[\theta](\psi, \xi_{[0,i+1]}, t)$ in a bottom-up fashion when a new signal point at $t=i+1$ is available. We propose SlidingC algorithm in the next paragraph to deal with the case when a node has a C operator. For lack of space, other cases with different operators can be found in ~\cite{2017monitor} and easily adapt to signal in uniform time domain as we defined. Note that we can use the RoSI of the root node to determine whether the formula is satisfied with the current partial signal. 

To handle the $\mathbf{C}_I^\tau$ operator, we introduce an algorithm SlidingC, presented in Algorithm~\ref{alg:SlidingC}.  It returns the $\tau$-th largest values in sliding windows, along with the times at which those values occur. 
Let $w = b-a+1$ denote the size of the sliding window and the signal length is $n$.
It maintains two heaps, the min heap and max heap. Min heap stores the top $\tau$ elements in the current window, ensuring the root element is the $\tau$-th largest. Max heap stores the remaining elements in the current window. Every time  the sliding window moves forward, the new signal value is pushed to the max heap if it's smaller than the root of min heap, otherwise it is pushed to min heap.
The old signal value is deleted lazily, for efficiency.  The outdated signal value is marked as outdated in a hash table and only actually deleted from the heap when it becomes the heap root. We move items between the two heaps as needed to ensure the min heap contains exactly $\tau$ elements. The $\tau$-th largest element at each step is then retrieved as the root of min heap, ensuring an optimal retrieval time of O(1). Finally we return $y = \{(i-a, \ctrobust(\mathbf{C}_{[a,b]}^{\tau} \xi>0, \xi_{[0,n]}, i-a))\}, i = 0,\ldots, a-b+n$. 
Since the insertion, deletion, and rebalance operations all have time complexity of $O(\log(\tau)+\log(w-\tau))$, the overall time complexity is $O(n (\log(\tau)+\log(w-\tau)))$.

\begin{algorithm}[tb]
    \caption{SlidingC $(\mathbf{C}_{[a,b]}^{\tau} \xi>0, \xi_{[0,n]}, t)$ }
    \label{alg:SlidingC}
    \begin{algorithmic}[1]
        \STATE \textbf{Output}: Sliding $\tau$-th largest values over $[0, n]$ with time $t$

        \STATE minHeap $\leftarrow \emptyset$ \hfill // min heap stores the largest $\tau$ elements within current window
        \STATE maxHeap $\leftarrow \emptyset$ \hfill // max heap stores remaining elements within current window
        \STATE result $\leftarrow \emptyset$ \hfill     // stores the $\tau$-th largest element per window

        \FOR{$i = 0$ \TO $n$}

            \STATE Insert $\xi(i) $ into the appropriate heap
            \IF{$i\geq b-a$}             
                \STATE Remove the outdated element from the appropriate heap
                \STATE Rebalance heaps to ensure minHeap contains exactly $\tau$ elements
                \STATE Append ($i-b$, minHeap.top()) to result \hfill // save the $\tau$-th largest value 
            \ENDIF
        \ENDFOR   
        
        \STATE return result
    \end{algorithmic}
\end{algorithm}

In Figure~\ref{a run of online monitoring}, we give an example of a run of the online monitoring Algorithm~\ref{alg:updateworklist} for CT-STL formula: $\varphi = \mathbf{G}_{[0,2]}\mathbf{C}_{[1,5]}^3 \xi>0$. We assume that at first the signal has five known points on $t=0,\ldots,4$, and we show the effect of receiving new signal values at $t=5, 6$ and $7$. The figure shows the states of the RoSI worklists at each node of the syntax tree at these times when monitoring the CT-STL formula presented in the figure. 
Each row in the table adjacent to a node shows the state of the worklists after we receive a new signal point at the time indicated in the first column.   
We can find that the RoSI shrinks or stays the same when new signal points are available. 
When we have the signal point at $t=6$, we can determine the RoSI at the root is $[7,7]$, i.e., robustness is $7$.
This shows $\varphi$ is satisfied without knowing the last signal point at $t=7$ in expected horizon.

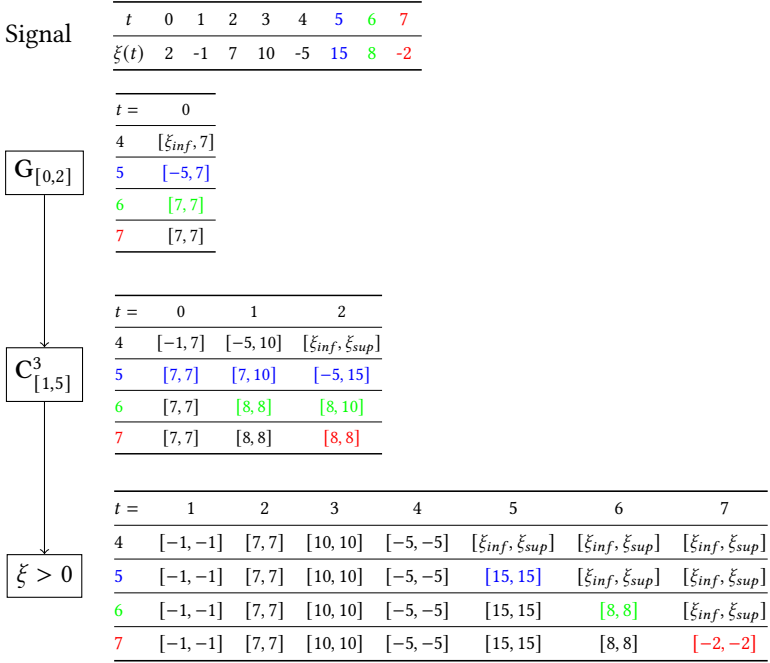
\begin{figure}[t]
\begin{tikzpicture}

    \node (n4) at (0,1.8)   {
    \resizebox{0.3\textwidth}{!}{%
    \begin{tabular}{@{}cccccccccccc@{}}
    \toprule
    $t$  & 0 & 1 & 2 & 3 & 4  & {\color{blue} 5} & {\color{green} 6} & {\color{red} 7} \\ \midrule
    $\xi(t)$ & 2  & -1  & 7  & 10  & -5  & {\color{blue} 15}  & {\color{green} 8} & {\color{red} -2} \\ 
    \bottomrule
    \end{tabular}
    } 
  };
  \node (n1) [draw, rectangle] at (-3,0)    {$\mathbf{G}_{[0,2]}$};
  \node (n2) [draw, rectangle, below=2cm of n1]    {$\mathbf{C}_{[1,5]}^3$};
  \node (n3) [draw, rectangle, below=2cm of n2]   {$\xi>0$};
\node (n8) [left=0.3cm of n4] {Signal};
  
    \node (n5) [right=0.3cm of n1]   {
    \resizebox{0.1\textwidth}{!}{%
    \begin{tabular}{@{}lc@{}}
    \toprule
    $t=$  & $0$ \\ \midrule
    $4$ &  $[\xi_{inf}, 7]$  \\ \midrule
    {\color{blue} $5$} & {\color{blue}$[-5,7]$}   \\ \midrule
    {\color{green} $6$} & {\color{green}$[7,7]$}    \\ \midrule
    {\color{red} $7$} & $[7,7]$ \\
    \bottomrule
    \end{tabular}
    } 
    };

    \node (n6) [right=0.3cm of n2]   {
    \resizebox{0.26\textwidth}{!}{%
    \begin{tabular}{@{}lccc@{}}
    \toprule
    $t= $ &  $0$ & $1$ & $2$\\ \midrule
    $4$ &  $[-1,7]$ & $[-5,10]$ & $[\xi_{inf}, \xi_{sup}]$  \\ \midrule
    {\color{blue} $5$} &  {\color{blue} $[7,7]$} & {\color{blue} $[7,10]$} & {\color{blue} $[-5, 15]$}  \\ \midrule
    {\color{green} $6$} &  $[7,7]$ & {\color{green}$[8,8]$} & {\color{green}$[8,10]$}  \\ \midrule
    {\color{red} $7$} & $[7,7]$ & $[8,8]$ & {\color{red} $[8,8]$}\\
    \bottomrule
    \end{tabular}
    }
    };

    \node (n7) [right=0.3cm of n3]   {
    \resizebox{0.63\textwidth}{!}{%
    \begin{tabular}{@{}lccccccc@{}}
    \toprule
    $t=$   &  $1$ & $2$ & $3$ & $4$ & $5$ & $6$ & $7$ \\ \midrule
    $4$  & $[-1,-1]$ & $[7,7]$ & $[10,10]$ & $[-5,-5]$ & $[\xi_{inf}, \xi_{sup}]$  & $[\xi_{inf}, \xi_{sup}]$ & $[\xi_{inf}, \xi_{sup}]$    \\ \midrule
    {\color{blue} $5$}  & $[-1,-1]$ & $[7,7]$ & $[10,10]$  &  $[-5,-5]$  & {\color{blue}$[15,15]$}  & $[\xi_{inf}, \xi_{sup}]$  & $[\xi_{inf}, \xi_{sup}]$  \\ \midrule
    {\color{green} $6$} & $[-1,-1]$ & $[7,7]$ & $[10,10]$ & $[-5,-5]$ & $[15,15]$ & {\color{green}$[8,8]$} & $[\xi_{inf}, \xi_{sup}]$  \\ \midrule
    {\color{red} $7$}  & $[-1,-1]$ & $[7,7]$ & $[10,10]$ & $[-5,-5]$ & $[15,15]$ & $[8,8]$ & {\color{red} $[-2,-2]$}\\
    \bottomrule
    \end{tabular}
    }
    };

  \draw[->] (n1.south)  -- (n2.north);
  \draw[->] (n2.south) -- (n3.north);
  
\end{tikzpicture}
    \caption{A snapshot of worklist[$v$] maintained for four (incremental) partial traces of the signal $\xi (t)$. $(\xi_{inf},\xi_{sup})$ correspond to the greatest lower bound and lowest upper bound on $\xi$. Each row is the worklist[$v$] when the signal point at the time in the first column is available. The colored entry is affected by the availability of a signal point with corresponding color.}
    \label{a run of online monitoring}
\end{figure}

\section{Experimental Results}\label{sec:experiment}

In this section, we demonstrate the expressivity of CT-STL and our online monitoring algorithm via two case studies.
Experiments were performed on an Intel Core i7-14700H CPU with 32GB of DDR5 RAM and a Windows 11 operating system. Our case studies have been implemented in Python; 
our implementation and case studies will be publicly available.

\subsection{Microgrid}
\label{case study microgrid}

\begin{figure}[t]
    \centering
    \includegraphics[width=\linewidth]{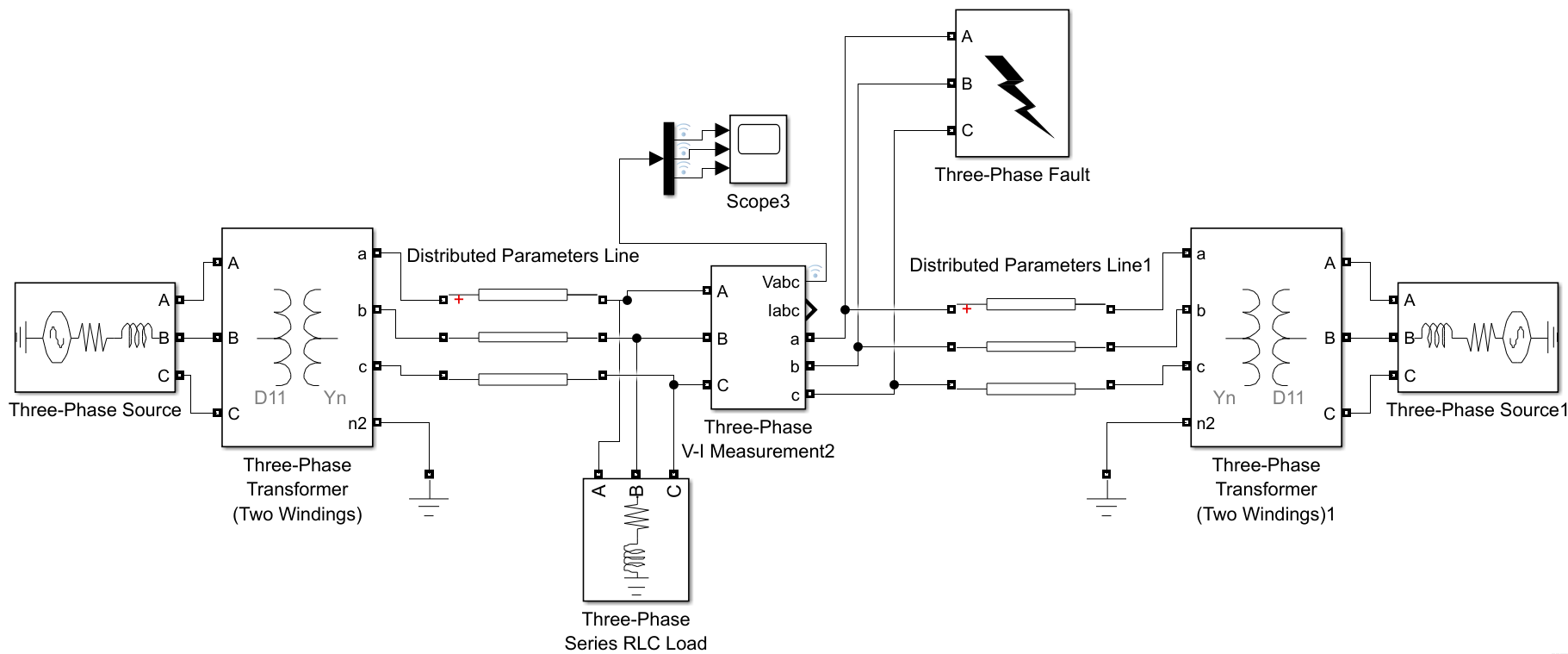}
    \caption{Microgrid model implemented using MATLAB Simulink.}
    \label{fig:microgrid simulink}
\end{figure}

\begin{figure}[t]
    \centering
    \subfloat[]{%
    \includegraphics[width=0.24\textwidth]{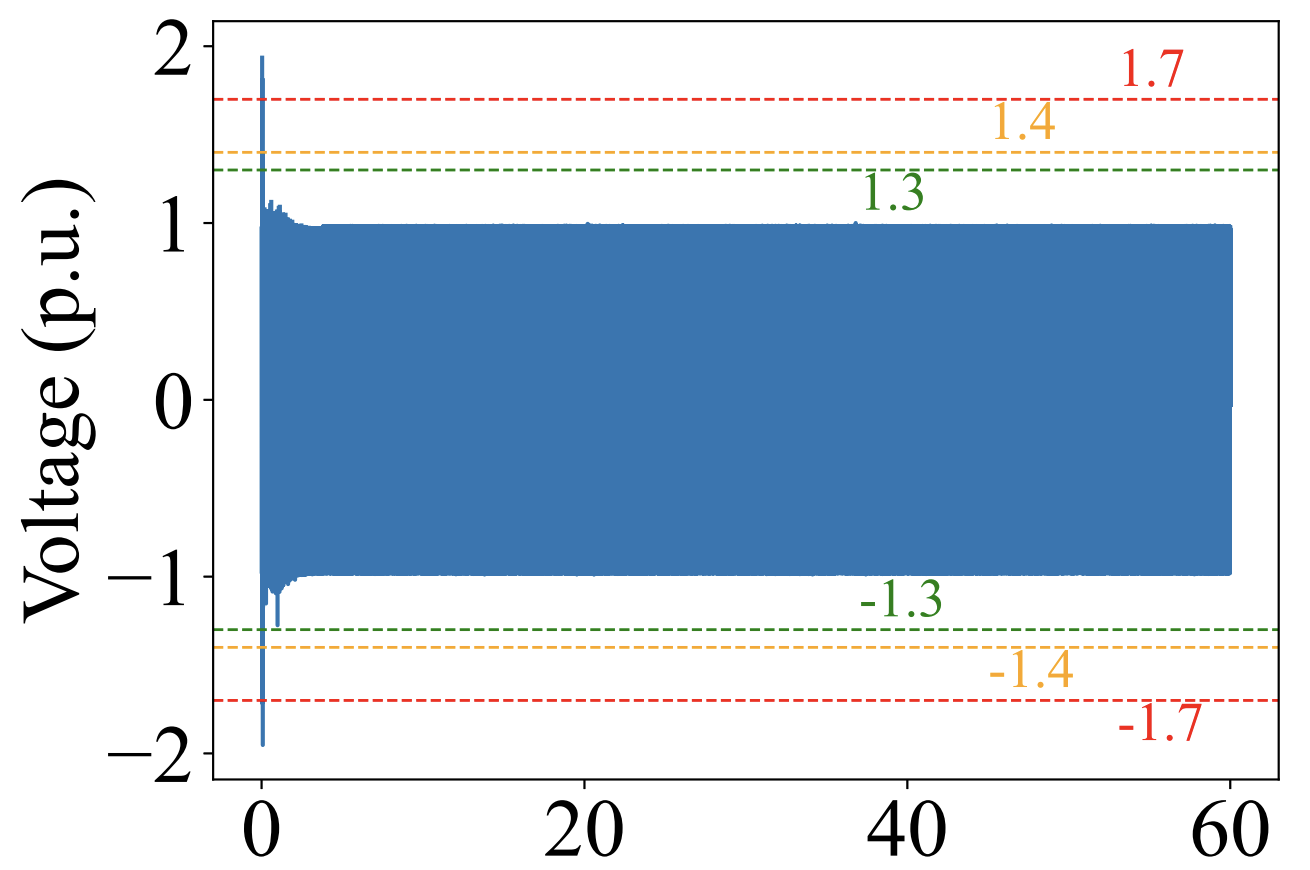}%
        \label{fig:v1}%
    }
    \hfill
        \subfloat[]{%
    \includegraphics[width=0.24\textwidth]{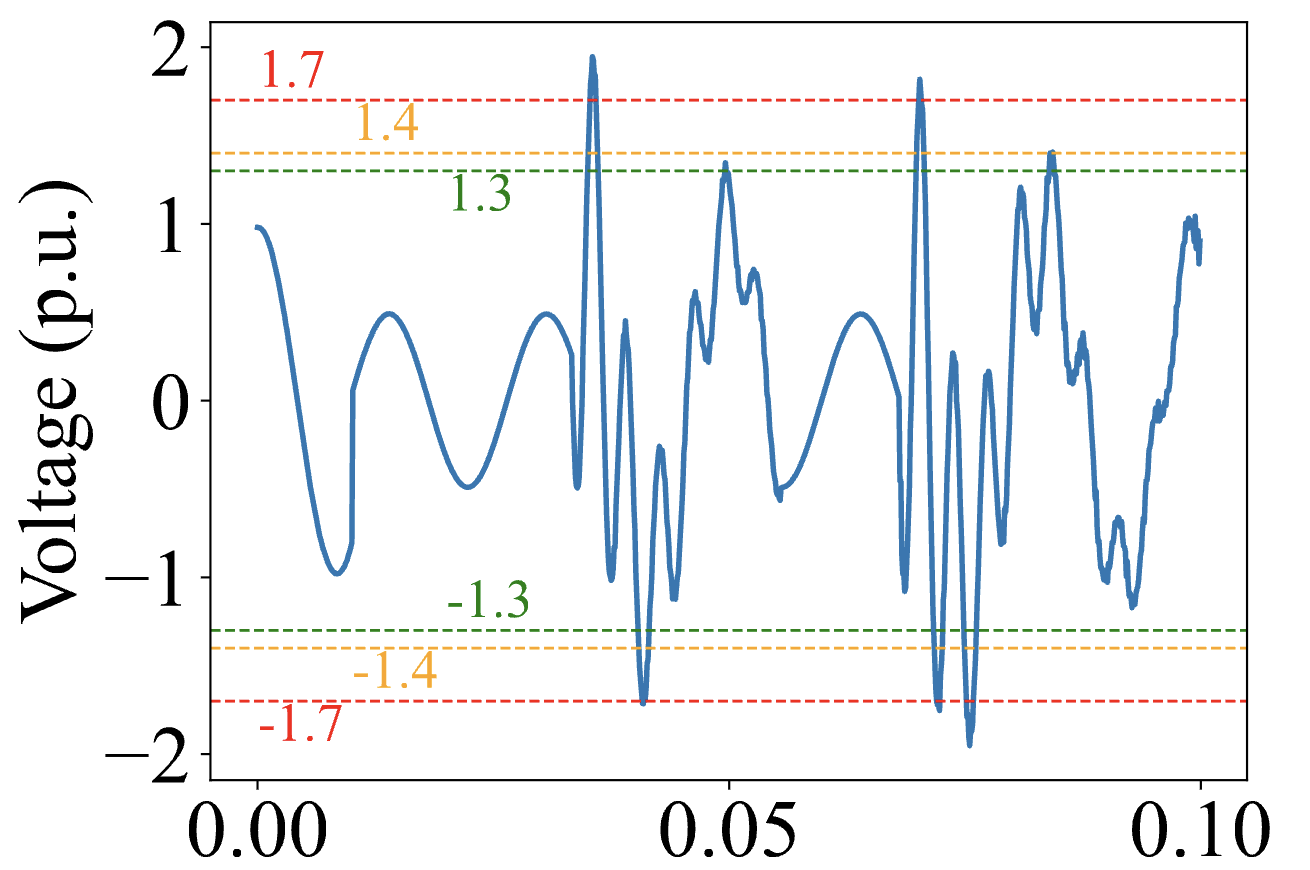}
        \label{fig:v2}
    }
    \hfill
    \subfloat[]{%
  \includegraphics[width=0.24\textwidth]{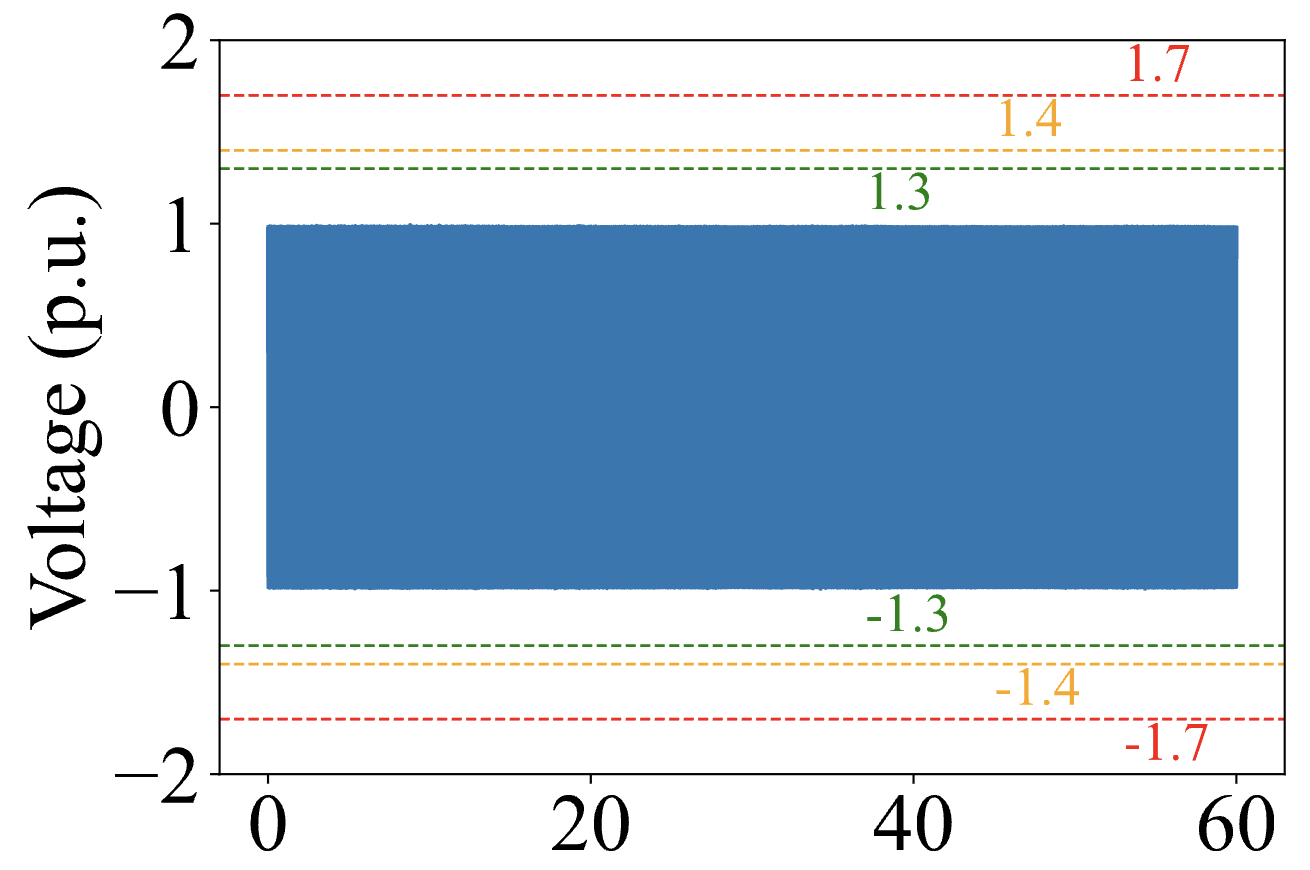}%
        \label{fig:v3}
    }
    \hfill
    \subfloat[]{%
 \includegraphics[width=0.24\textwidth]{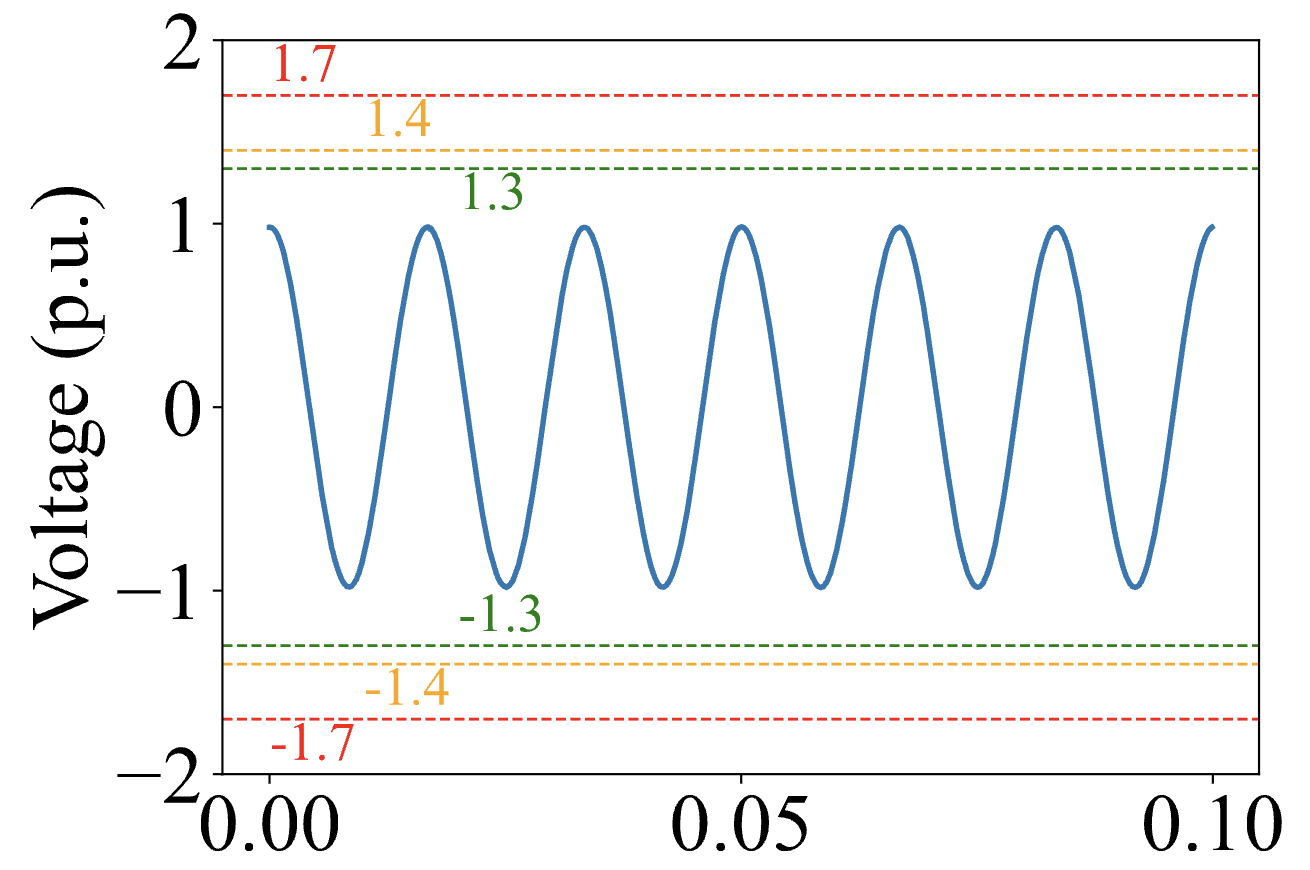}%
        \label{fig:v4}
    }
    \caption{Simulated traces of voltages. (a)~Voltage signal that violates $\varphi_5$. (b)~The first 0.1 s of Figure~\ref{fig:voltages}(a) determines the violation. (c)~~Voltage signal that satisfies $\varphi_5$. (d)~The first 0.1 s of Figure~\ref{fig:voltages}(c) shows no violation; so does the remaining signal.}
    \label{fig:voltages}
\end{figure}

\begin{figure}[t]
    \centering
    \subfloat[]{%
        \includegraphics[width=0.24\textwidth]{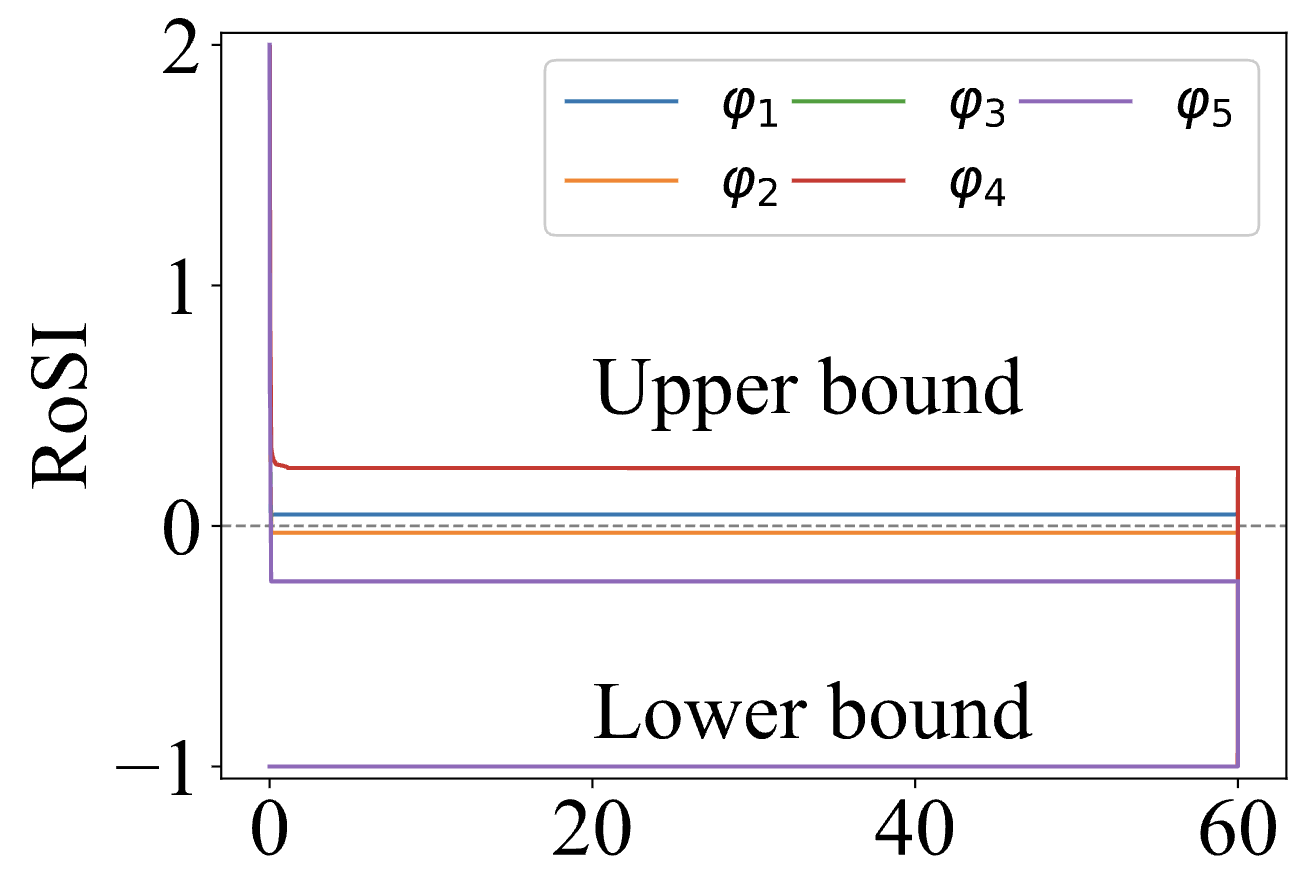}%
        \label{fig:rosi1}%
    }
    \hfill
        \subfloat[]{%
        \includegraphics[width=0.24\textwidth]{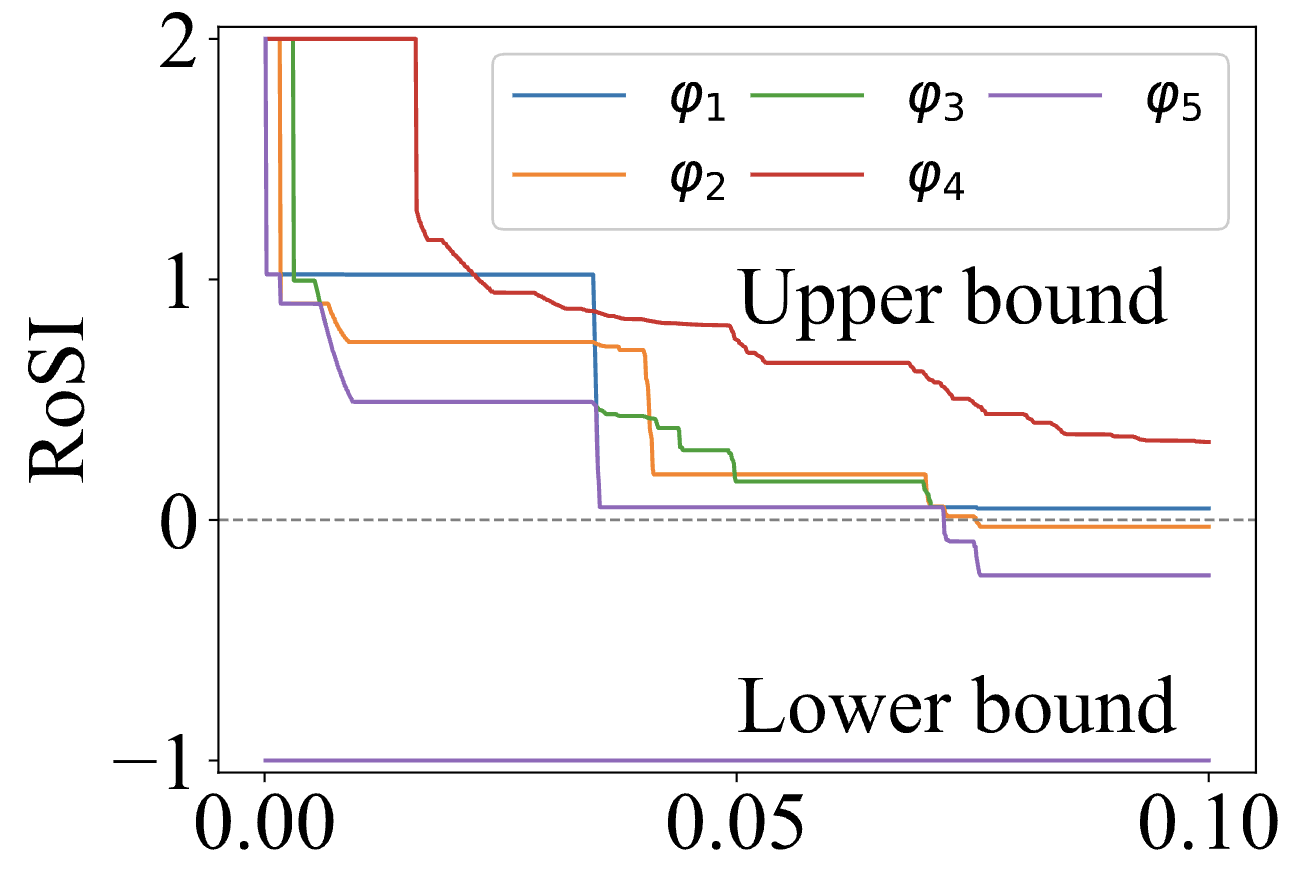}%
        \label{fig:rosi2}%
    }
    \hfill
    \subfloat[]{%
        \includegraphics[width=0.24\textwidth]{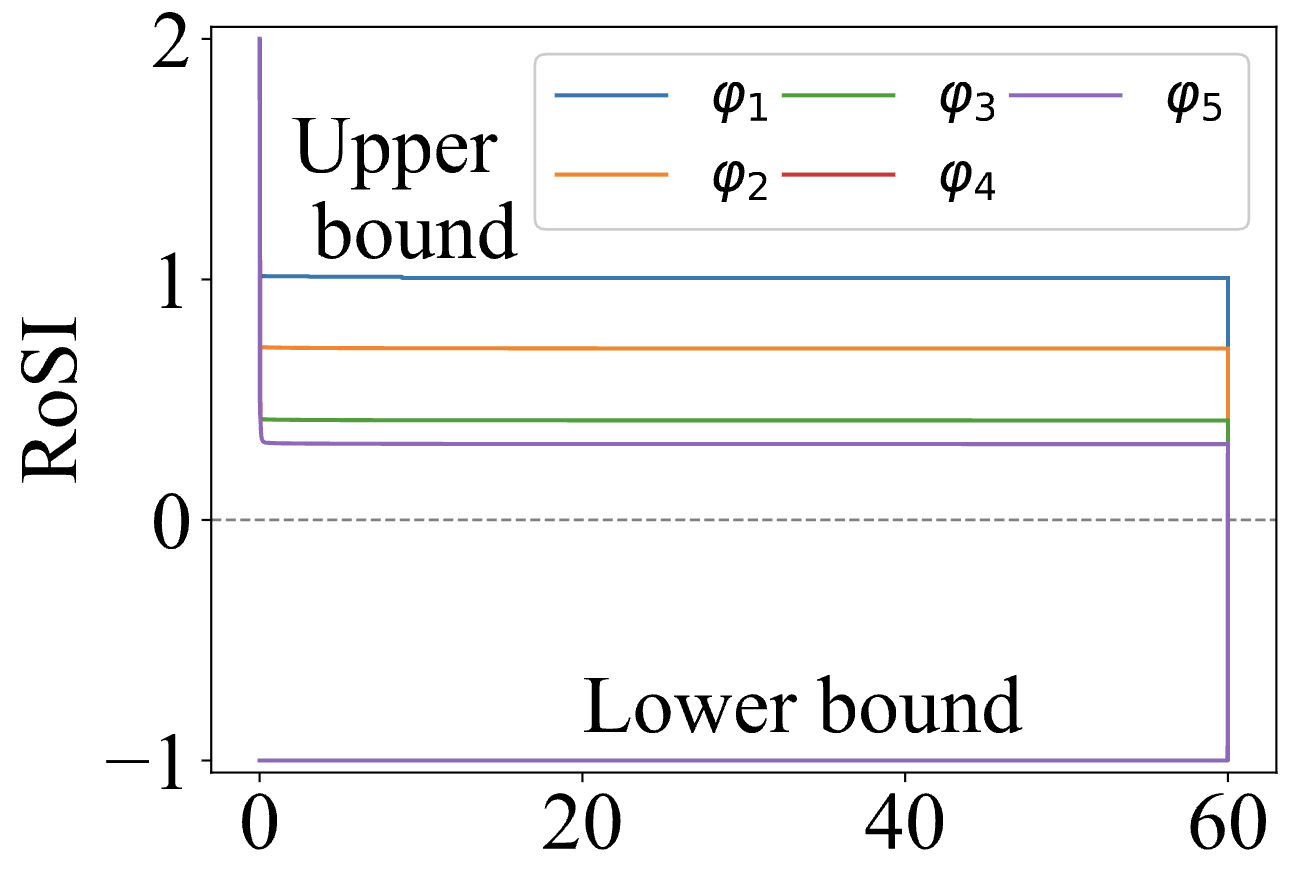}%
        \label{fig:rosi3}%
    }
    \hfill
        \subfloat[]{%
        \includegraphics[width=0.24\textwidth]{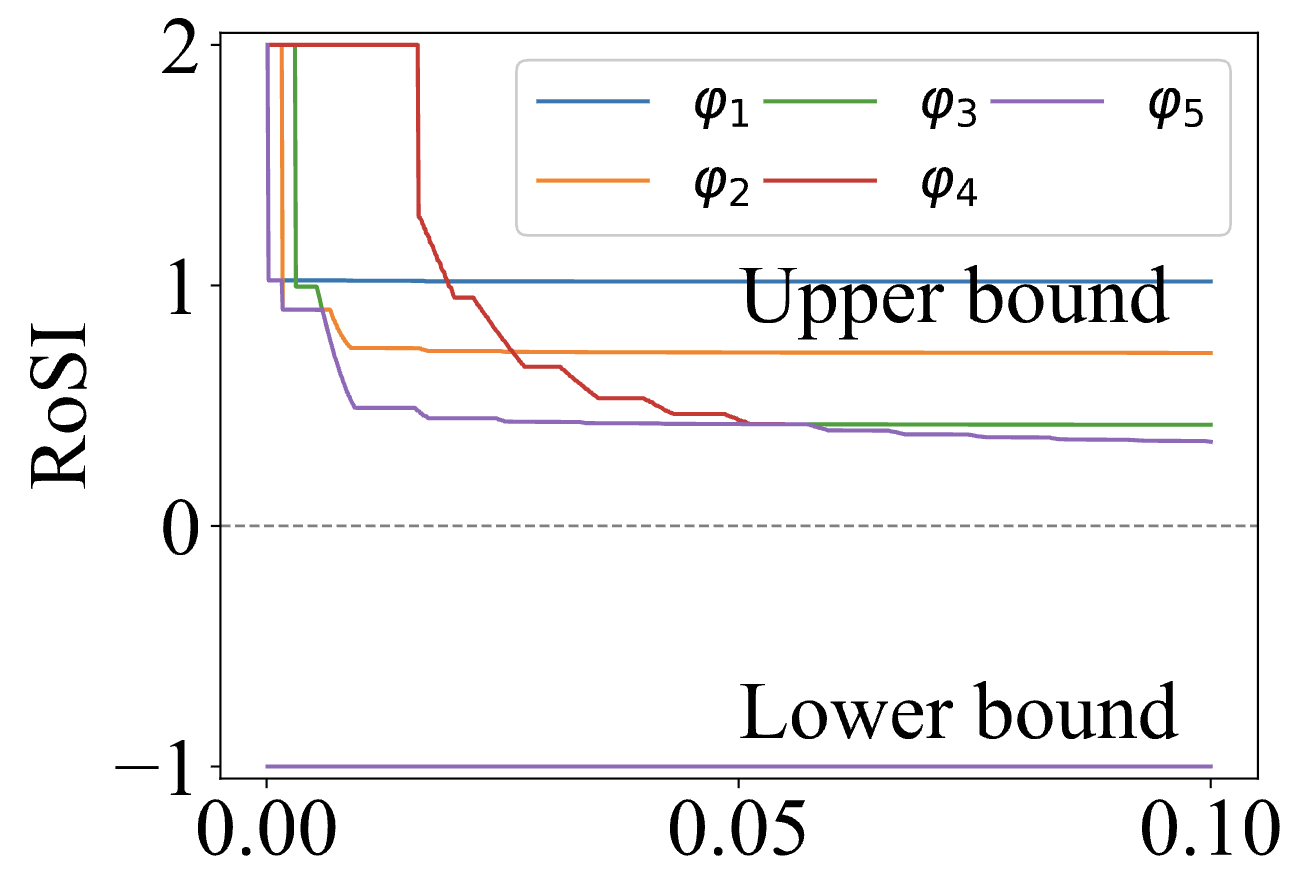}%
        \label{fig:rosi4}%
    }
    \caption{RoSIs for $\varphi_1$ to $\varphi_5$ are analyzed for Figures~\ref{fig:voltages}(a)-(d).
    (a)~The RoSI of $\varphi_5$ is sufficient to determine the violation of $\varphi_5$ w.r.t.\ the signal.
    (b)~The upper bound of the RoSI of $\varphi_5$ turns  negative shortly after $t=0.07$, allowing for an early termination.
    (c)~The lower bound of the RoSI of $\varphi_5$ becomes positive near the end of the trace. This occurs because the rest short time is sufficient for $\varphi_5$ to be violated.
    (d)~The lower bound of the RoSI for $\varphi_5$ becomes positive near the end of the trace. This occurs because a very short time window is sufficient for $\varphi_5$ to be violated.}
    \label{fig:rosi}
\end{figure}

We study the microgrid case presented in Example~\ref{example:1}. 
We implement a microgrid model in MATLAB Simulink\textregistered, as shown in Figure~\ref{fig:microgrid simulink}. 
The system consists of two three-phase sources and two transformers located on the left and right sides.
The three-phase sources generate 10.5 kV, 60 Hz power, with the transformer having a ratio of 10.5 kV/110 kV.
Power is supplied to an RLC load through a distributed parameter line.
A three-phase fault module is added on the load side, in the center, to simulate short circuits.
The acceptable region defines the maximum cumulative durations permitted for different levels of instantaneous overvoltage within any one-minute time window. This acceptable region can be described using a CT-STL formula $\varphi$ as follows.
\begin{align*}
    \varphi_1 &=\; v \leq 2 \\
    \varphi_2 &=\; \mathbf{C}_{[0,600000]}^{600000-16}\,v < 1.7 \\
    \varphi_3&=\; \mathbf{C}_{[0,600000]}^{600000-30}\,v < 1.4 \\
    \varphi_4&=\; \mathbf{C}_{[0,600000]}^{600000-160}\,v < 1.3 \; \\
    \varphi_5 &= \mathbf{G}_{[0,600000]}\,\;(\varphi_1 \wedge \varphi_2 \wedge \varphi_3 \wedge \varphi_4) 
\end{align*}


\begin{table}[h]
    \centering
    \begin{tabular}{c c c c c c c}
        \toprule
        \textbf{Requirement} & \textbf{Num.} & \textbf{Early} & \multicolumn{2}{c}{\textbf{Simulation time (s)}} & \multicolumn{2}{c}{\textbf{Overheads(s)}}\\
        \cmidrule(lr){4-5} \cmidrule(lr){6-7}
        & \textbf{Traces} & \textbf{Termination} & \textbf{Offline} & \textbf{Online} & \textbf{Naive} &\textbf{Algorithm 2}\\
        \midrule
        $\varphi_1$ & 100 & 90 & 5370.05 & 2916.37 & $1.32\times10^4$ &10.82\\
        $\varphi_2$ & 100 & 98 & 5370.05 & 3619.59 & $1.54\times 10^4$&129.04\\
        $\varphi_3$ & 100 & 99   & 5370.05 & 3699.36 &$1.57\times 10^4$&132.29\\
        $\varphi_4$ & 100 & 99 & 5370.05 & 5325.05 &$2.73\times10^4$&197.62\\
        $\varphi_5$ & 100 & 92 & 5370.05 & 2975.03 &$5.86\times10^4$&200.45\\
        \bottomrule
    \end{tabular}
    \caption{Experimental Evaluation of CT-STL online monitoring on microgrid.}
    \label{tab:mc monitor time}
\end{table}


First, two B-C phase short circuits are set between $t=0.01-0.033\,s$ and 0$t=.055-0.068\,s$, respectively. 
Figures~\ref{fig:voltages}(a) and (b) show the first 1 minute and the first 0.1 s of an example phase-B voltage signal, respectively, \emph{with} the short circuit.
Additionally, Figures~\ref{fig:voltages}(c) and (d) show the first 1 minute and the first 0.1 s of an example phase-B voltage signal, respectively, \emph{without} the short circuit.
The RoSIs of these examples at $t=0$, computed using the online monitoring Algorithm~\ref{alg:updateworklist}, are shown in Figure~\ref{fig:rosi}.
Note that we set the maximum interval robustness to $[-1,2]$ for presentation purposes.
We can see that $\varphi_5$ is the conjunction of $\varphi_1$ to $\varphi_4$, so its RoSI is the interval arithmetic minimum of the RoSIs of $\varphi_1$ to $\varphi_4$. 
Furthermore, the robust interval reduces as time advances and our online monitoring algorithm can determine whether CT-STL robustness semantics is positive or negative with partial signals. 
For example, Figure~\ref{fig:rosi}(b) demonstrates the upper bounds of the RoSIs of $\varphi_2, \varphi_3,\varphi_5$ becoming negative at $t= 0.0754\,s, 0.072\,s, 0.072\,s$, respectively. 
This is because, between $t= 0-0.0754\,s$, the cumulative time where the absolute voltage value exceeds 1.7 p.u.\ is more than 1.6\,ms, thus $\varphi_2$ is violated and our online monitoring algorithm can decide an early termination. 
Figures~\ref{fig:rosi}(c) shows no overvoltage over the interval $[0,600000]$. 
However, the satisfaction $\varphi_1$ to $\varphi_5$ can only be conclusively determined near the traces' completion because the cumulative-time thresholds represent less than $0.03\%$ of the total trace duration.
Even a few remaining unsimulated points could potentially violate the requirements.

We evaluate our online monitoring algorithm, as shown in Table~\ref{tab:mc monitor time}. 
We simulate 100 instances of 60\,s traces with time steps of 0.1\,ms.
The online algorithm can end a simulation earlier (either by detecting a violation or by identifying a definitive robust satisfaction interval), resulting in substantial time savings.
Table~\ref{tab:mc monitor time} demonstrates that online monitoring saves up to $46\%$ simulation time~($>30\%$ in a majority of cases).
We also implement a naive online monitoring algorithm that use a modification of the offline monitoring algorithm to recursively compute the robust satisfaction intervals as defined in Definition~\ref{def:rosi}.
The online algorithm exhibits less computational overhead compared with the naive online approach consistently by a factor of 118-1200x.

\subsection{Artificial Pancreas}

The artificial pancreas is a system to provide insulin therapy for Type 1 diabetes (T1D) patients. 
It consists of an insulin infusion pump and a subcutaneous Continuous Glucose Monitor (CGM) to measure blood glucose (BG) levels. We use $t_{hypo}, t_{hyper}, t_{eu}$ to denote the percentage of time when a patient's BG is in hypoglycemia (i.e, BG<70 mg/dL), hyperglycemia (i.e, BG >180 mg/dL) and euglycemia (i.e, 70$\leq$ BG $\leq$ 180 mg/dL) over 24 hours, respectively. For most adults, the therapy goal is to make $t_{hypo}<4\%, t_{hyper} < 25\%, t_{eu} > 70 \%$~\cite{glycemicgoals}. We use $x$ to denote the BG measurements signal. Then we specify our control purpose using CT-STL formulas as follows.

\begin{align*}
    \phi_1 = & \;\neg \mathbf{C}_{[0,T_1]}^{0.04\times T_1} (x<70) \\
    \phi_2 = & \;\neg \mathbf{C}_{[0,T_1]}^{0.25\times T_1} (x>180) \\
    \phi_3 = & \;\mathbf{C}_{[0,T_1]}^{0.7\times T_1} (x\geq 70 \wedge x\leq 180) \\
    \phi_4 = & \; \phi_1 \wedge \phi_2 \wedge \phi_3 \\
    \phi_5 = & \; \mathbf{G}_{[0,T_2]} \phi_4
\end{align*}

These formulas specify least cumulative time each BG requirement holds. For example, $\varphi_3$ specifies the cumulative time that BG maintains in the euglycemia state should be more than $70\%$ of $T_1$. 
We set the monitor horizon as 24 hours~\cite{glycemicgoals}, i.e., $T_1=24*60/5 =288$ with a sampling rate of 5 minutes. 
We set $T_2 = 36$ (3 hours) and 72 (6 hours).

\begin{figure}[t]
    \centering
    \subfloat[]{%
\includegraphics[width=0.38\textwidth]{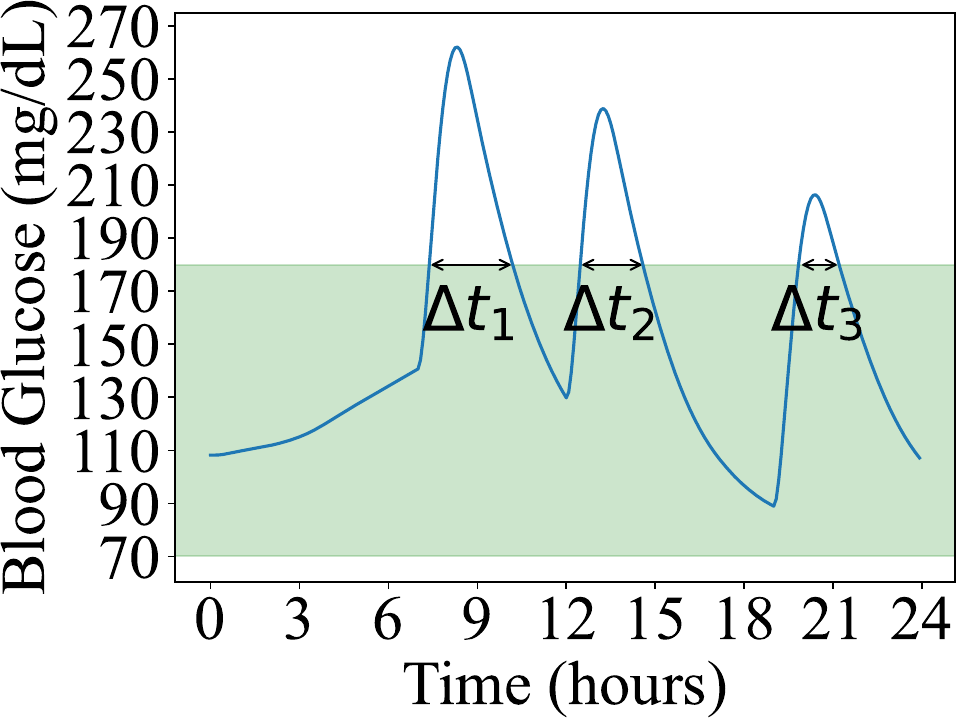}%
        \label{fig:ap1}
    }
    \hspace{4em}
    \subfloat[]{%
    \includegraphics[width=0.38\textwidth]{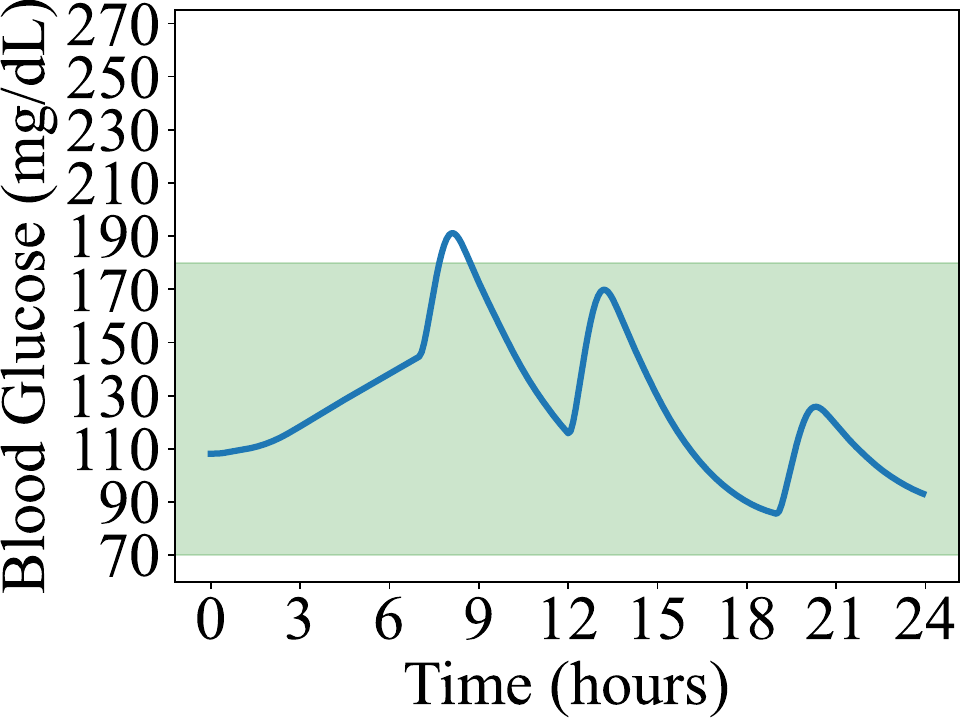}%
        \label{fig:ap2}
    }
    
    \caption{Example blood glucose measurements in 24 hours. (a)~a trace with three hyperglycemia areas. (b)~a trace with only one short hyperglycemia area.}
    \label{fig:ap measurements}
\end{figure}

\begin{figure}[t]
    \centering
    \subfloat[]{
    \includegraphics[width=0.38\linewidth]{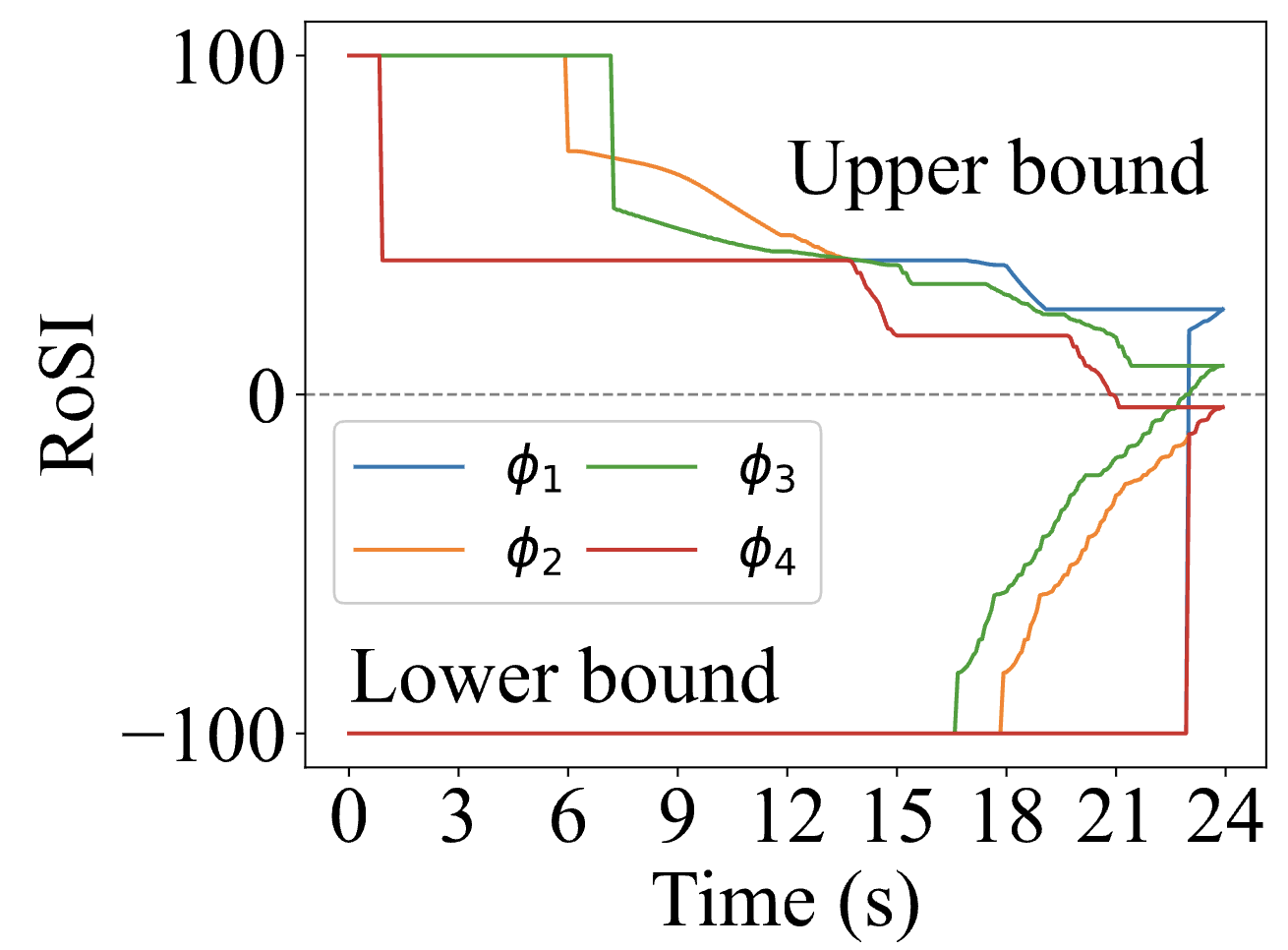}
        \label{fig:ap rosi1}
    }
    \hspace{4em}
    \subfloat[]{
    \includegraphics[width=0.38\linewidth]{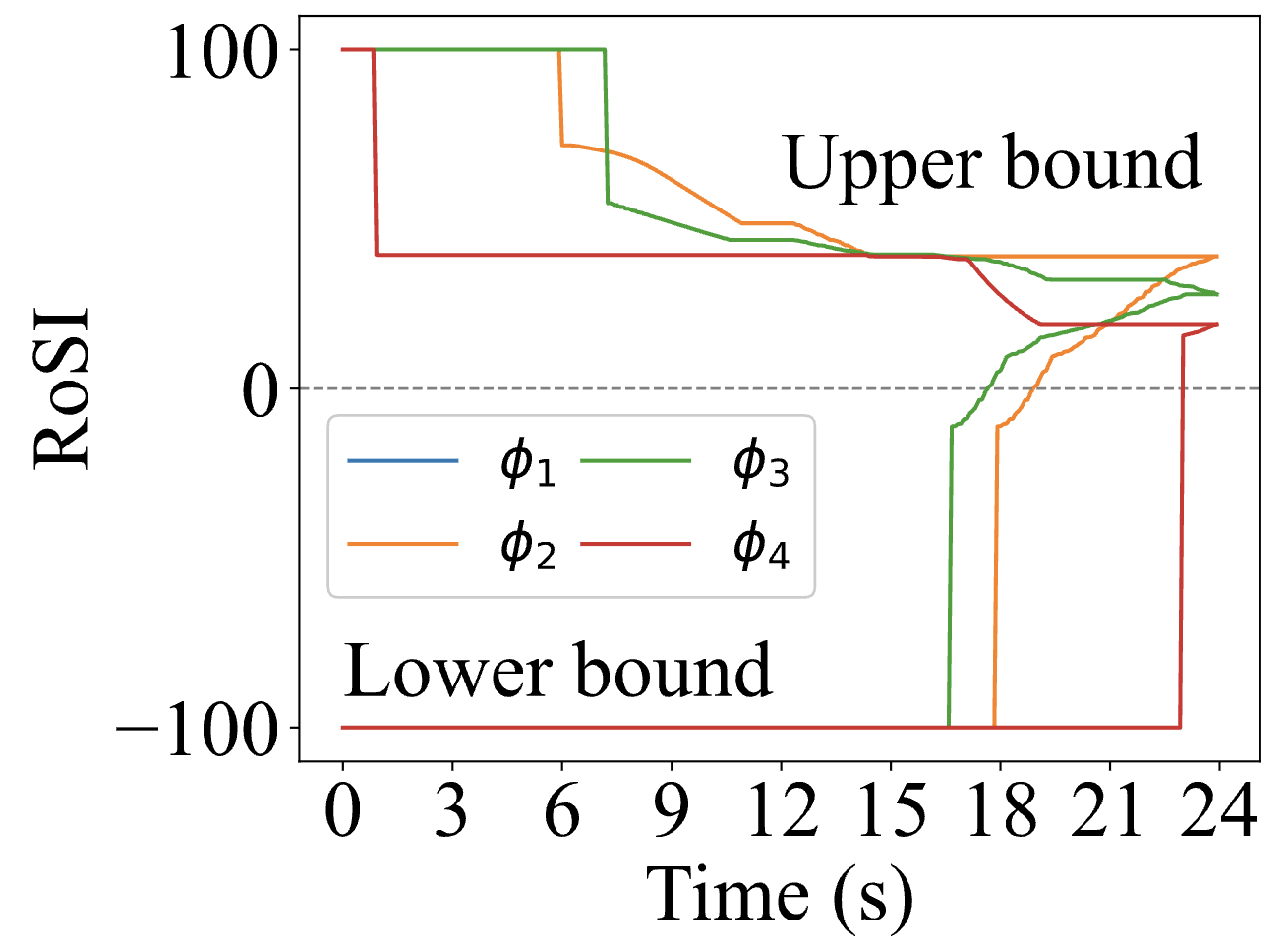}
        \label{fig:ap rosi2}
    }
    \caption{The RoSIs of $\phi_1,\phi_2,\phi_3,\phi_4$ for signals in Figure~\ref{fig:ap measurements}.}
    \label{fig:ap rosi}
\end{figure}

\begin{table}[h]
    \centering
    \begin{tabular}{c c c c c c c}
        \toprule
        \textbf{Requirement} & \textbf{Num.} & \textbf{Early} & \multicolumn{2}{c}{\textbf{Simulation time (s)}} & \multicolumn{2}{c}{\textbf{Overheads (s)}}\\
        \cmidrule(lr){4-5} \cmidrule{6-7}
        & \textbf{Traces} & \textbf{Termination} & \textbf{Offline} & \textbf{Online} &\textbf{Naive} & \textbf{Algorithm 2}\\
        \midrule
        $\phi_1,T_1=288$ & 100 & 100 & 11.62 &  11.14 & 302.24&30.34\\
        $\phi_2,T_1=288$ & 100 & 94 & 11.62 & 9.84 & 325.58&32.50\\
        $\phi_3,T_1=288$ & 100 & 99   & 11.62 & 10.15 & 313.27&28.61\\
        $\phi_4,T_1=288$ & 100 & 94 & 11.62 & 10.14 & 932.72&92.22\\
        $\phi_5,T_1=288, T_2=36$ & 100 & 100 & 13.08 & 10.23 & 4544.69 &178.14\\
        $\phi_5,T_1=288, T_2=72$ & 100 & 100 & 14.53 & 10.23 & 8947.08 &235.70 \\
        $\phi_5,T_1=144, T_2=36$ & 100 & 100 & 7.26 & 5.35 & 1097.28 &58.43 \\
        \bottomrule
    \end{tabular}
    \caption{Experimental Evaluation of CT-STL online monitoring on artificial pancreas.}
    \label{tab:ap sim time}
\end{table}

We use Simulink{\textregistered} model~\cite{APMatlab} to simulate the BG levels in 24 hours, as shown in Figure~\ref{fig:ap measurements}. We adjust the meal intake amount and get two traces. Figure~\ref{fig:ap measurements}(a) shows a trace with three hyperglycemia areas while Figure~\ref{fig:ap measurements}(b) shows a trace with only one short hyperglycemia area. In Figure~\ref{fig:ap1}, $\Delta t_1=170 \,\text{min}, \Delta t_2=130 \,\text{min}, \Delta t_3=80 \,\text{min}$.

We use online monitoring algorithm to calculate the RoSIs and Figure~\ref{fig:ap rosi} demonstrates the results. 
Figures~\ref{fig:ap rosi}(a) and~\ref{fig:ap rosi}(b) are the RoSIs of $\varphi_1$ to $\varphi_5$ at $t=0$ w.r.t.\ the signals in Figures~\ref{fig:ap measurements}(a) and~\ref{fig:ap measurements}(b), respectively. 
In Figure~\ref{fig:ap rosi1}, we find the upper bounds of $\phi_2, \phi_4$ become negative at around $t=20:55$. This is because there are already more than $25\%\times24\,\text{h}=360\,\text{min}$ cumulative time when the BG is above 180\,mg/dL. This happens because $\Delta t_1 + \Delta t_2 = 300\, \text{min} < 360 \,\text{min}$ while $\Delta t_1 + \Delta t_2 + \Delta t_3 = 380 \,\text{min} > 360\, \text{min}$. 
For $\phi_1$ and $\phi_3$, their robustness becomes positive when the remaining signal is not long enough to violate the hypoglycemia requirements. In Figure~\ref{fig:ap rosi}(b), since most signal values are in the euglycemia area, all RoSIs converges to positive values.
We also conduct evaluation for the online monitoring algorithm as shown in Table~\ref{tab:ap sim time}. We simulate 100 traces of the BG level with different meal intake and applied insulin. The trace length is equal to the horizon of each formula. For each trace, we run the online monitoring Algorithm~\ref{alg:updateworklist} until early termination or the end of each trace.
Our online monitoring algorithm saves an average of $17.5\%$ of simulation time compared to the offline approach. 
Table~\ref{tab:ap sim time} also demonstrates that, compared to the naive online algorithm, our online algorithm has less overhead consistently by a factor of 10-25x. 

\section{Conclusion}\label{sec:conclusion}
We present a new specification language for cyber-physical systems that extends STL with a new cumulative temporal operator comparing with a threshold $\tau$ for how many discrete time steps its nested formula is satisfied in a closed interval of time.  This new operator, together with the standard STL ones, allows the expression of more complex properties such as microgrid energy consumption requirements.
We also extend the original notion of STL quantitative semantics to address the addition of the new cumulative temporal operator.  The definition of quantitative semantics for the cumulative temporal operator returns the $\lceil\tau\rceil$-th largest value of robustness measured in the interval. We show that this quantitative semantics is both sound and complete with the qualitative one. We also provide procedures to monitor this new logic both online and offline, and evaluate these monitoring approaches in two case studies.  We envision many exciting future research directions. For example, we plan to study several related problems, such as satisfiability and consistency checking for CT-STL requirements, testing and falsification analysis, failure explanation, and specification mining techniques for this new logic.

\bibliographystyle{ACM-Reference-Format}
\bibliography{hongkai}

\end{document}